\documentclass[review,sort&compress]{elsarticle}
\usepackage{lineno,hyperref}
\usepackage{floatrow}
\usepackage{graphicx}
\usepackage{caption}

\usepackage{pdfpages}
\usepackage{multicol}        
\usepackage[bottom]{footmisc}
\usepackage{subfig}
\usepackage{amsfonts}
\usepackage{hyperref}
\usepackage[cmex10]{amsmath}
\usepackage{amsmath}
\usepackage{amsthm}
\usepackage{xcolor}
\usepackage{mathtools}

\newfloatcommand{capbtabbox}{table}[][\FBwidth]
\newtheorem{thm}{Theorem}
\newtheorem{lemma}{Lemma}
\newtheorem{corollary}{Corollary}
\newtheorem{definition}{Definition}
\theoremstyle{remark}
\newtheorem{remark}{Remark}

\usepackage{lipsum}

\usepackage{epstopdf}
\usepackage{algorithmic}

\usepackage{enumitem}
\setlist[enumerate]{leftmargin=.5in}
\setlist[itemize]{leftmargin=.5in}

\usepackage[left=2cm, right=2cm, top=2cm, bottom=2cm]{geometry}

\usepackage{bm}
\usepackage[normalem]{ulem}
\usepackage{xcolor}

\usepackage{blindtext}

\usepackage{amssymb}

\begin{document}

\begin{frontmatter}



\title{A Parrondo paradox in susceptible-infectious-susceptible dynamics over periodic temporal networks}

\author[affiliation_1]{Maisha Islam Sejunti}
\author[affiliation_2]{Dane Taylor}
\author[affiliation_1,affiliation_3,affiliation_4]{Naoki Masuda\corref{mycorrespondingauthor1}}
\ead{naokimas@buffalo.edu}

\address[affiliation_1]{Department of Mathematics, State University of New York at Buffalo, NY, 14260-2900, USA}
\address[affiliation_2]{School of Computing and Department of Mathematics and Statistics, University of Wyoming, Laramie, WY, 82071-3036,USA}
\address[affiliation_3]{Institute for Artificial Intelligence and Data Science, State University of NewYork at Buffalo, NY, 14260-5030, USA}
\address[affiliation_4]{Center for Computational Social Science, Kobe University, Kobe 657-8501, Japan}

\begin{abstract}
Many social and biological networks periodically change over time with daily, weekly, and other cycles. Thus motivated, we formulate and analyze  susceptible-infectious-susceptible (SIS) epidemic models over temporal networks with periodic schedules. More specifically, we assume that the temporal network consists of a cycle of alternately used static networks, each with a given duration. 
We observe a phenomenon  in which two static networks are individually above the epidemic threshold but the alternating network composed of them renders the dynamics below the epidemic threshold, which we refer to as a Parrondo paradox for epidemics. We find that network  structure plays an important role in shaping this phenomenon, and we study its dependence on the connectivity between and number of subpopulations in the network. We associate such paradoxical behavior with anti-phase oscillatory dynamics of the number of infectious individuals in different subpopulations.
\end{abstract}


\end{frontmatter}


\section{Introduction}\label{sec:intro}


The connectivity of social contact networks is a main determinant for how contagions, i.e., spreading of communicative diseases, opinions, rumor, ideas, and so forth, occur in a population.
Epidemic modeling on networks has been a useful tool that can help us understand, predict, and mitigate contagious processes occurring in the physical world and online \cite{keeling2005networks, easley2010networks,danon2011networks,miller2014epidemic,pastor2015epidemic,kiss2017mathematics}. In fact, the structure of social and transportation networks can also dynamically change on the timescale of disease spreading. For example, human and animal mobility is a main obvious reason why contact networks vary over time.
%
%
Airport transportation networks, which consist of airports as nodes and direct commercial flights between two airports as edges, also alter over days, weeks, and longer time \cite{rocha2017dynamics,sugishita2021recurrence}.
Time-varying networks are collectively called temporal networks, and epidemic processes on temporal networks have been extensively studied \cite{bansal2010dynamic,danon2011networks,holme2012temporal,blonder2012temporal,masuda2013predicting,holme2015modern,masuda2017temporal,holme2019temporal,masuda2020guide,masuda2021concurrency}. Ignoring 
variations of a network over time may lead to wrong or inaccurate conclusions about dynamics on networks.  For example, in the SIR process, the number of eventually infected nodes, called the final size, tends to be systematically different between simulations on a given temporal network and those on a counterpart static network (i.e., the aggregated, that is, time-averaged static network corresponding to the temporal network) that ignores temporal nature of the original network, even if the infection and recovery rates are assumed to be the same between the two sets of simulations (e.g., \cite{masuda2013predicting, masuda2021concurrency}). Therefore, simulations on the static network would over-or under-estimate epidemic spread occurring on the given temporal network.

One strategy for mathematically and computationally modeling temporal networks and dynamical processes on them is to use periodic temporal networks. Periodic temporal networks are useful because they allow us to derive and analyze a mapping of the given system's state from one time point to one period after, reminiscent of the Poincar\'{e} map. A useful subclass of periodic temporal networks is periodic switching temporal networks, which are defined as time-varying networks in which one static network is used for a certain duration and then the network switches to another network, which is used for a certain duration, and so forth, and we come back to the first static network after the last one to complete one period, and then repeat the same cycle. It is common to use each static network for the same duration for simplicity. Periodic temporal networks in discrete time are by definition periodic switching temporal networks because each one static network is present in each time step, which generally changes over discrete times. Popularity of using periodic switching temporal networks is presumably due to their analytical tractability. Periodically switching networks have been applied to modeling, e.g., contagion processes 
\cite{gomez2010discrete,valdano2015analytical,speidel2016temporal,onaga2017concurrency,zino2021analysis,somers2023sparse,allen2024compressing} 
evolutionary dynamics~\cite{olfati2007evolutionary,kun2009evolution,li2020evolution,guan2021structural,sheng2023evolutionary,su2023strategy,bhaumik2023fixation,belykh2014evolving,sheng2024strategy}, gene regulation \cite{edwards2000analysis,jenkins2013temporal}, network control \cite{hou2016structural,li2017fundamental,guan2019controllability,hou2022time}, and random walks \cite{perra2012random,starnini2012random,rocha2014random,petit2019classes}. 

A phenomenon that can occur in periodically switching dynamical systems or stochastic processes is the Parrondo paradox. Originally, the Parrondo paradox refers to the situation in which a combination of two losing strategies generates a winning strategy \cite{harmer1999losing,parrondo2000new,harmer2002review}.
When a person alternately and repetitively plays two different losing gambling games (i.e., losing regardless of the player's strategy), the player may still win the game, even on average.  The Parrondo paradox has also been found in theoretically motivated dynamical systems \cite{almeida2005can,canovas2006dynamic,canovas2013revisiting,danca2014generalized,danca2015parrondo,mendoza2018parrondo,cima2020dynamic}, social dynamics models \cite{ye2016effects,lai2020social,ye2020impact,ye2021effects}, evolutionary dynamics on  networks \cite{ye2013multi}, financial investment models \cite{spurgin2005switching,chakrabarti2014switching}, chemical engineering \cite{osipovitch2009systems}, molecular transport \cite{heath2002discrete}, 
dynamics of allele frequency~\cite{reed2007two,abbott2010asymmetry,cheong2019paradoxical}, cancer biology~\cite{capp2021does}, cellular biology \cite{cheong2018multicellular}, and so on. This phenomenon is also known in control theory community; a switching dynamical system composed of alternation of unstable dynamical systems can become stable \cite{liberzon2003switching,yang2009stabilization,yang2011stabilization,yang2014survey,xiang2014stabilization,lu2020stabilization,yang2020exponential}. Models of epidemic dynamics are no exception~\cite{masuda2004subcritical,danca2012parrondo,cheong2020relieving,page202125,danca2023controlling}.
%
%
For example, Cheong \textit{et al.}\,considered a model of COVID-19 outbreak \cite{cheong2020relieving}.
In their model, when the community is not under lockdown, the infection rate is higher and the hospital cost is increased for individual persons. In contrast, when a lockdown is implemented, the total economic costs are increased despite the rate of epidemic spreading being reduced. They argue that implementing either one of these two strategies is a losing strategy. However, switching between the two strategies can reduce the overall cost per day; Cheong \textit{et al.}\, interpreted this to be a winning strategy, thereby representing a Parrondo paradox~\cite{cheong2020relieving}. As another example, in \cite{danca2023controlling}, the authors implemented the Parrondo paradox idea on a periodic parameter switching algorithm with the aim of controlling the COVID-19 outbreak in their model. Thus, the literature suggests that the Parrondo paradox is useful for suppressing epidemic spreading.

In the present study, we exploit this idea for potentially suppressing epidemic spreading over periodically switching networks with explicit network structure to investigate the effects of such structure. Specifically, we study the Parrondo paradox occurring on the susceptible-infectious-susceptible (SIS) dynamics over periodically switching temporal networks. We first find that the Parrondo paradox often occurs for small networks (i.e., those with a few nodes) modeling interaction between a small number of subpopulations of individuals. We corroborate our results with analysis of larger networks having the same structure of interaction between subpopulations. We characterize the observed paradoxical dynamics by quantifying anti-phase oscillation between the time courses of the fraction of infectious individuals in different subpopulations. We finally develop a perturbation theory, which also assists in explaining the reason for the paradoxical dynamics. Code used in this paper is available on GitHub~\cite{githubSejunti}.

\section{Preliminaries\label{sec:back}}

We first review the SIS epidemic model over static networks, definition of temporal networks, and the SIS model over periodic temporal networks. 

\subsection{SIS model over static networks}\label{sec:SIS_static}

In this section, we present the SIS model over static networks and its epidemic threshold.
We consider a static weighted network on $N$ nodes, $G(\mathcal{V} ,\mathcal{E})$, that may have self-loops;
$\mathcal{V} = \{v_1, \ldots, v_N \}$ is the set of nodes, and $\mathcal{E}\subset  \mathcal{V}\times\mathcal{V}$ is the set of edges. We denote the weighted adjacency matrix by $\textbf{A} = (a_{ij})$. By assumption, each node takes one of the two states, i.e., susceptible (denoted by S) or infectious (denoted by I), at any given time $t \in \mathbb{R}$. If node $v_i$ is susceptible and its neighbor $v_j$ is infectious, then $v_j$ infects $v_i$ at a constant rate $\beta a_{ij}$, where $\beta$ ($>0$) is the infection rate. If $v_i$ is adjacent to multiple infectious nodes, each infectious neighbor infects $v_i$ independently of each other. In addition, any infectious node recovers to the susceptible state at a constant recovery rate, denoted by $\mu$ ($>0$). Therefore, the SIS model is a Markov process in continuous time with $2^N$ states.
 
The binary random variable, $X_i(t)$, encodes the state of $v_i$, i.e., $X_i(t) = 0$ and $X_i(t) = 1$ if $v_i$ is susceptible or infectious, respectively. Let $x_i(t)$ be the probability that the node $v_i$ is infectious at time $t$.
Because there is no closed system of ODEs that exactly describes the evolution of 
$x_i(t)$ (with $i \in \{1, \ldots, N\}$), a system of ODEs approximating the evolution of $x_i(t)$, called
the individual-based-approximation (IBA), has been studied~\cite{pastor2015epidemic}. The IBA assumes independence of $\{ x_1(t), \ldots, x_N(t) \}$. Under the IBA, the SIS dynamics is given by \cite{van2008virus,van2009performance,van2012epidemic}
\begin{equation}
\frac{\text{d}x_i(t)}{\text{d}t}=\beta \left[ 1-x_i(t) \right] \sum_{j=1}^N a_{ij}x_j(t)-\mu x_i(t), \quad i\in \{1, \ldots, N \}.
\label{nonlinear_final_SIS}
\end{equation}
Equation~\eqref{nonlinear_final_SIS} can be expressed in the succinct form as
\begin{equation}
\frac{\text{d}\textbf{x}(t)}{\text{d}t}=\left[\beta \textbf{A} -\beta \text{diag}(x_i(t))\textbf{A} - \mu \textbf{I}\right]\textbf{x}(t),
 \label{nonlinear_matrix_SIS}
\end{equation}
where $\textbf{x}(t)=[x_1(t),\dots,x_N(t)]^{\top}$, the transposition is denoted by ${}^{\top}$,
$\text{diag}(x_i(t))$ represents the diagonal matrix with diagonal entries $x_1(t)$, $\ldots$, $x_N(t)$, and
$\textbf{I}$ is the $N\times N$ identity matrix.

The linearized ODE for Eq.~\eqref{nonlinear_matrix_SIS} under the assumption that $x_i(t)\approx 0$, $\forall$ $i \in \{1, \ldots, N \}$, or that the infection rate is posed such that epidemic dynamics is near the disease-free equilibrium, is given by
\begin{equation}
    \frac{\text{d}\textbf{x}(t)}{\text{d}t} =(\beta\textbf{A}-\mu \textbf{I})\textbf{x}(t) \equiv \textbf{M} \textbf{x}(t).
\label{linearizesismodel} 
\end{equation}
Equation~\eqref{linearizesismodel} leads to
\begin{equation}
\textbf{x}(t) = e^{\textbf{M}t}\textbf{x}(0)= e^{(\beta \textbf{A} - \mu \textbf{I})t}\textbf{x}(0).
\label{eq:linear-solution}
\end{equation}
It is known that Eq.~\eqref{eq:linear-solution} provides an upper bound of the fraction of infectious nodes in the original stochastic SIS dynamics. Specifically, the probability that the infection persists in the network is upper-bounded by \cite{duff1966differential,ganesh2005effect,van2008virus,ogura2016stability}
\begin{equation}
Pr[X(t) \neq 0] 
\leq \sqrt{N ||X(0)||_1} e^{ \lambda_{\max}(\textbf{M}) t }
= \sqrt{N ||X(0)||_1} e^{(\beta \lambda_{\max}(\textbf{A})- \mu)t },
\label{bound_of_solution}
\end{equation}
where $X(t) = (X_1(t), \ldots, X_N(t))$, $\left\|X(0)\right\|_1=\sum_{i=1}^N X_i(0)$ represents the number of initially infectious nodes, and $\lambda_{\max}$ represents the largest eigenvalue of a matrix in modulus. Note that the Perron-Frobenius theorem guarantees that $\lambda_{\max}(\textbf{A})$ is real and positive; the theorem holds true because the entries in $\textbf{A}$ are nonnegative and $\textbf{A}$ is irreducible.

Under the IBA, the epidemic grows if and only if \cite{van2008virus,ogura2016stability}
\begin{equation}
\lambda_{\max}(\textbf{M}) >0 .
\label{epidemic_threshold_M}
\end{equation}
Since $\lambda_{\max}( \textbf{M}) = \beta\lambda_{\max}(\textbf{A})-\mu$, Eq.~\eqref{epidemic_threshold_M} is equivalent to $\beta\lambda_{\max}(\textbf{A})>\mu$, or
$\frac{\beta}{\mu}>\frac{1}{\lambda_{\max}(\textbf{A})}.$
Therefore, under the IBA, the epidemic threshold, i.e., the value of the infection rate above which the epidemic grows, is given by $\beta_{\text{c}} = \mu / \lambda_{\max}(\textbf{A})$.

Another interpretation of Eqs.~\eqref{nonlinear_final_SIS}, \eqref{nonlinear_matrix_SIS}, \eqref{linearizesismodel}, \eqref{eq:linear-solution}, and \eqref{epidemic_threshold_M} is that a node represents a subpopulation of individuals \cite{krause2018stochastic,vizuete2020graphon,wang2021suppressing}.
In this case, the network is that of subpopulations, and $x_i(t)$ represents the fraction of infectious individuals in the $i$th subpopulation at time $t$.
Figure~\ref{fig:schemetic_figure_of_SBM} illustrates a network of individuals that are divided into two equal-sized subpopulations (with $10^3$ individuals each). The natural coarse graining of this network yields a $2\times 2$ matrix $\textbf{A}$, which encodes connectivity between the two subpopulations and within each subpopulation. To allow this interpretation, we allow the network to have self-loops (i.e., positive diagonal entries of $\textbf{A}$).

\begin{figure}[t]
\begin{center}
\includegraphics[scale=.40]{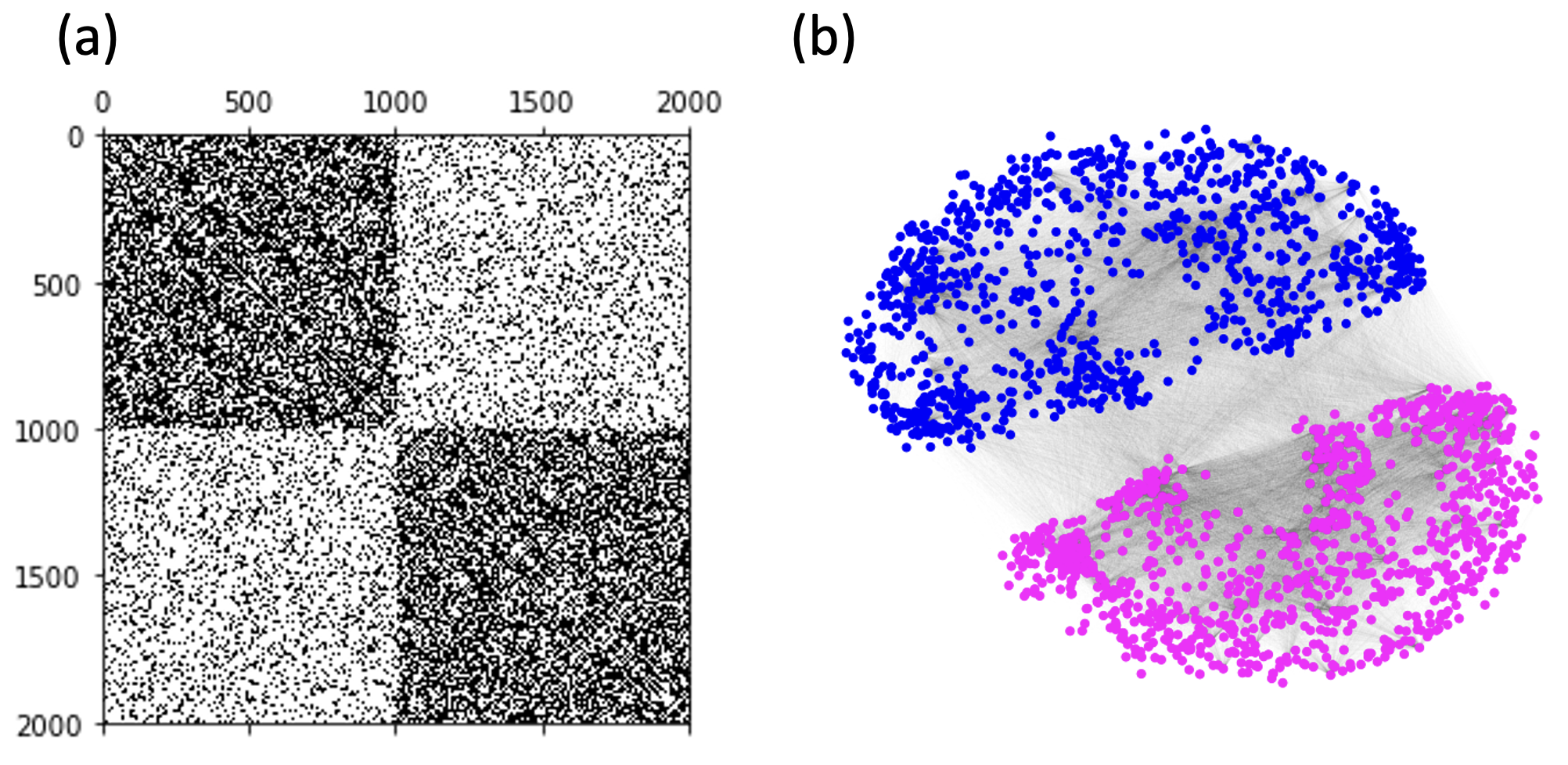}
\caption{A network with $N=2000$ individuals organized in two subpopulations. 
(a) Adjacency matrix. Black dots indicate edges in the network, which are encoded by nonzero $A_{ij}$ entries. By grouping the first $10^3$ nodes into a subpopulation and the remaining $10^3$ nodes into another, we obtain a network composed of two subpopulations with self-loops. The subpopulation structure can be summarized by considering the probabilities of connection within and between subpopulations (which can be stored in a $2 \times 2$ matrix).
(b) Network visualization. Nodes (e.g., people) in the first and second subpopulation are colored blue and magenta, respectively.
}
\label{fig:schemetic_figure_of_SBM}
\end{center}
\end{figure}


\subsection{Periodic switching temporal networks} \label{sec:temp} 

The time-varying $N\times N$  weighted adjacency matrix of the temporal network is denoted by $\textbf{A}(t)= (a_{ij}(t))$, where $a_{ij}(t)$ represents the weight of the edge from node $v_i$ to $v_j$ at time $t \in \mathbb{R}$.
When $\textbf{A}(t+T) = \textbf{A}(t)$ for any $t$, we say that the temporal network is periodic and that $T$ is the period.

We consider a special case of periodic temporal networks in which $\textbf{A}(t)$ switches between a finite number of static adjacency matrices according to a periodic schedule. We call such a periodic temporal network a periodic switching (temporal) network. The periodic switching network is defined by a sequence of adjacency matrices $\{\textbf{A}^{(1)},\textbf{A}^{(2)}, \ldots,\textbf{A}^{(\ell)}\}$ and the duration of each $\ell'$th static network denoted by $\tau_{\ell'}$ (with $\ell' \in \{ 1, \ldots, \ell \}$). The periodic switching network is defined by
\begin{equation}
\textbf{A}(t) = \begin{cases}
\textbf{A}^{(1)} & (nT \le t < nT + \tau_1),\\
\textbf{A}^{(2)} & (nT + \tau_1 \le t < nT + \tau_1 + \tau_2),\\
\; \; \vdots & \quad \quad \quad \vdots\\
\textbf{A}^{(\ell)} & (nT + \tau_1 + \cdots + \tau_{\ell-1} \le t < (n+1)T),
 \end{cases}
\end{equation}
where $n \in \mathbb{Z}$. Note that $T = \sum_{\ell' = 1}^{\ell} \tau_{\ell'}$.

In Fig.~\ref{fig:examplenetwork}, we visualize a periodic switching network with $\ell=2$ that obeys a weekly cycle (i.e., period $T=7$ days) and contains $N=2$ nodes. In each cycle, the first adjacency matrix, $\textbf{A}^{(1)}$, is used in the first five days (i.e., $\tau_1 = 5$ days), corresponding to the weekdays, and then the second adjacency matrix, $\textbf{A}^{(2)}$, is used in the next two days (i.e., $\tau_2 = 2$ days), corresponding to the weekend. Each node of this periodic switching network may represent a group of people, and there are different disease transmission rates within and between members of two groups, and at different times (i.e., weekday versus weekend).

\begin{figure}[t]
\begin{center}
\includegraphics[scale=0.4]{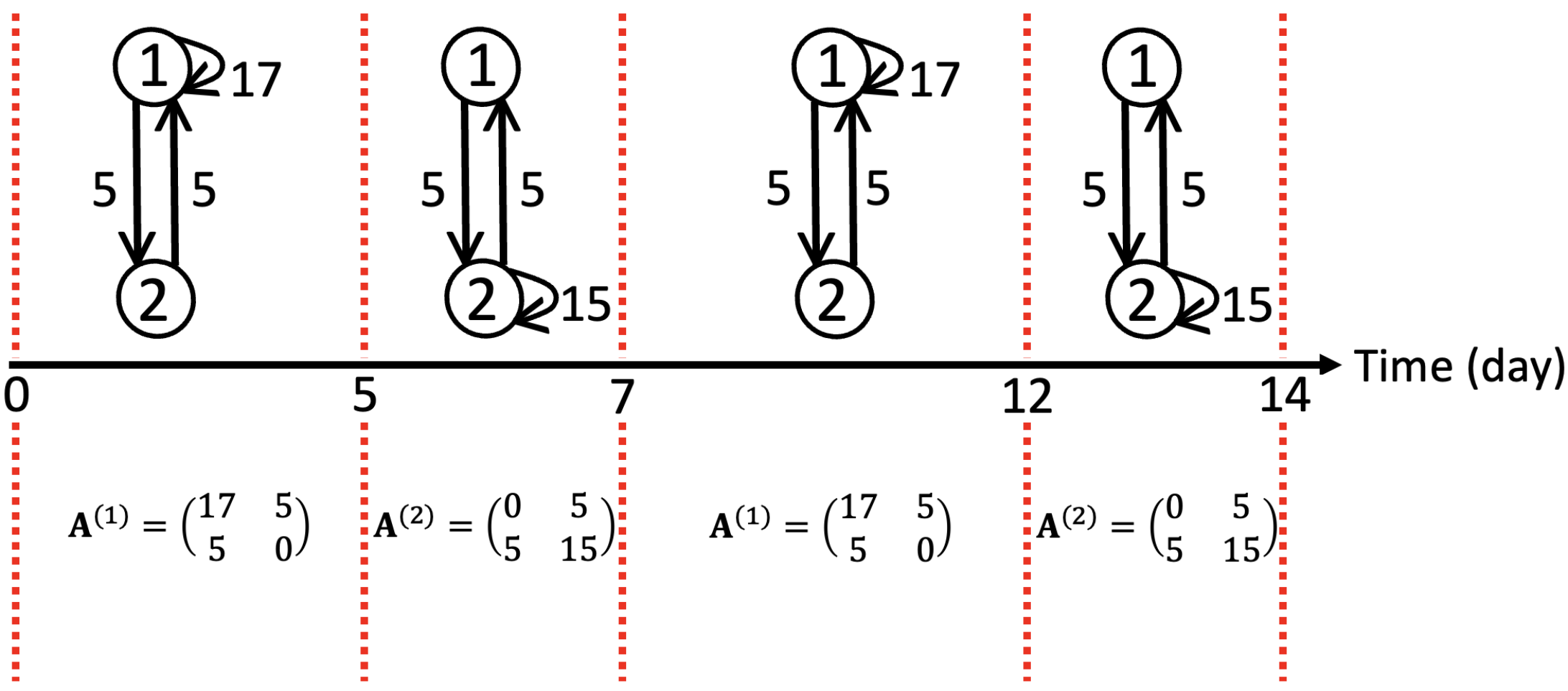}
\caption{A periodically switching network with $N=2$ nodes, a period of $T=7$ days, and two sub-intervals of lengths $\tau_1=5$, corresponding to weekdays, and $\tau_2=2$, corresponding to weekends.
}
\label{fig:examplenetwork}
\end{center}
\end{figure}

\subsection{SIS model over periodic temporal networks\label{sec:SIS_temporal}}

We now consider the SIS model in continuous time on periodic temporal networks. First of all,
the IBA for temporal networks in continuous time is given by \cite{allen2024compressing} 
\begin{equation}
  \frac{\text{d}x_i(t)}{\text{d}t} =\beta \left[ 1-x_i(t) \right] \sum_{j=1}^N a_{ij}(t) x_j(t)-\mu x_i(t).
  \label{eq:IBA-temporal-general}
\end{equation}
Around the disease-free equilibrium, i.e., ${\bf x}(t)={\bf 0}$, the linearized SIS dynamics (which assumes small $x_i(t)$ values) is given by
\begin{equation}
    \frac{\text{d}\textbf{x}(t)}{\text{d}t}= \left[ \beta\textbf{A}(t) -\mu \textbf{I} \right] \textbf{x}(t) \equiv \textbf{M}(t) \textbf{x}(t).
\label{eq:linearize_general_periodic_temporal_network} 
\end{equation}
Equation~\eqref{eq:linearize_general_periodic_temporal_network} yields
\begin{equation}
\textbf{x}(T)=e^{\int_0^T \textbf{M}(t)\text{d}t} \textbf{x}(0) \equiv \mathcal{T} \textbf{x}(0).
\label{eq:x(T)_general_periodic_temporal_network}
\end{equation}

To prove that $\lambda_{\max}(\mathcal{T})$ is real and positive, we show that $\mathcal{T}$ is a matrix of which all the entries are positive under the assumption that the time-averaged network is strongly connected so that its adjacency matrix is nonnegative and irreducible.
To this end, we start with
\begin{equation}
\mathcal{T}
%
%
=e^{\int_0^T (\beta \textbf{A}(t)-\mu \textbf{I})\text{d}t}
=e^{\beta\textbf{E}(T) -\mu T \textbf{I} },
\label{matrix_exponential}
\end{equation}
where $\textbf{E}(T)= (e_{ij}(T)) =  \int_0^T \textbf{A}(t)\text{d}t$. 
(Note that one has ${\bf E}(T) = \sum_{\ell' = 1}^{\ell} \tau_{\ell'} {\bf A}^{(\ell')}$ for a periodic switching network.) Because the identity matrix commutes with any matrix, $\textbf{E}(T)$ commutes with the identity matrix. 
Therefore, we obtain \cite{hall2013lie}
\begin{equation}
\mathcal{T}
%
%
=e^{\beta\textbf{E}(T)}\times e^{-\mu T \textbf{I} }
=e^{\beta\textbf{E}(T)}\times e^{-\mu T }\textbf{I}
=e^{\beta\textbf{E}(T)}\times e^{-\mu T }.
\label{proof_of_commute}
\end{equation}
Because $\beta>0$ and all the entries in $\textbf{A}(t)$ are nonnegative, all the entries of $\textbf{E}(T)$ are nonnegative. Let us assume that $\textbf{E}(T)$ is an irreducible matrix (i.e., $ [\textbf{E}(T)^k]_{ij} > 0$ $\forall i, j$ for some integer $k$). This condition is equivalent to assuming that the time-averaged network is strongly connected. The irreducibility and nonnegativity of $\textbf{E}(T)$ implies that $e^{\beta\textbf{E}(T)}$ is a strictly positive matrix (i.e., $[e^{\beta \textbf{E}(T)}]_{ij}>0$ $\forall i,j$)
%
\cite{hall2013lie, meyer2023matrix}. We show the proof in \ref{lemma_proof}. This observation, combined with Eq.~\eqref{proof_of_commute} and $e^{-\mu T } > 0$, guarantees that $\mathcal{T}$ a strictly positive matrix. Therefore, $\lambda_{\max}(\mathcal{T})$ is real and positive owing to the  Perron-Frobenius theorem. 

The epidemic threshold is given by the value of $\beta$ at which $\lambda_{\max}(\mathcal{T})=1$. Eigenvalue $\lambda_{\max}(\mathcal{T})$ is equivalent to the largest Floquet multiplier in Floquet theory (see \ref{sec:flo} for a brief review of a Floquet theory for periodic linear dynamical systems). 
We also define $\lambda_\text{F} \equiv T^{-1} \ln \lambda_{\max}(\mathcal{T})$, which gives the average growth rate of the infection. Under the IBA, the epidemic grows on a periodic temporal network if and only if 
\begin{equation}
\lambda_{\text{F}} >0, 
\label{floquet_criterion}
\end{equation} 
thereby generalizing Eq.~\eqref{epidemic_threshold_M} for the temporal network case.
Note that
$\lambda_\text{F} > 0$ is equivalent to $\lambda_{\max}(\mathcal{T})>1$ and vice versa.
Because $\lambda_\text{F}$ is the largest of the Floquet exponents in Floquet theory (see \ref{sec:flo}), we refer to $\lambda_\text{F}$ as the largest Floquet exponent in the following text.

In the case of periodic switching networks, Eq.~\eqref{eq:IBA-temporal-general} is reduced to
\cite{speidel2016temporal,valdano2018epidemic}
\begin{equation}
  \frac{\text{d}x_i(t)}{\text{d}t} =\beta (1-x_i(t)) \sum_{j=1}^N a^{(\ell')}_{ij}x_j(t)-\mu x_i(t),
  \label{nonlinear_final_SIS_temporal}
\end{equation}
where $a^{(\ell')}_{ij}$ is the $(i, j)$ entry of $\textbf{A}^{(\ell')}$ , and $\ell' \in \{ 1, \ldots, \ell \}$ is the unique value satisfying
$t\in [nT + \tau_1 + \cdots + \tau_{\ell'-1}, nT + \tau_1 + \cdots + \tau_{\ell'-1} + \tau_{\ell'})$, $\exists n \in \mathbb{Z}$. Linearization of Eq.~\eqref{nonlinear_final_SIS_temporal} around the disease-free equilibrium yields
\begin{equation}
    \frac{\text{d}\textbf{x}(t)}{\text{d}t}=(\beta\textbf{A}^{(\ell')}-\mu \textbf{I})\textbf{x}(t)
\label{linearize_sismodel_conts_temporal} 
\end{equation}
with the same value of $\ell'$ as that in Eq.~\eqref{nonlinear_final_SIS_temporal}. Equation~\eqref{linearize_sismodel_conts_temporal} gives
Eq.~\eqref{eq:x(T)_general_periodic_temporal_network} with
\begin{equation}
\mathcal{T} = e^{(\beta \textbf{A}^{(\ell)}-\mu I)\tau_{\ell}} \cdots e^{(\beta \textbf{A}^{(1)}-\mu I)\tau_1}.
\label{linearize_sismodel_conts_temporal_new} 
\end{equation}
Note that $\mathcal{T} = e^{\beta \textbf{A}^{(\ell)}\tau_{\ell}} \cdots e^{\beta \textbf{A}^{(1)}\tau_1} e^{- \mu T}$ because any $e^{- \mu I \tau_{\ell'}}$ $(= e^{- \mu \tau_{\ell'}} I)$ commutes with any matrix.

\section{Results}\label{sec:exps}

In this section, we investigate the Parrondo paradox for the SIS dynamics over periodic switching networks. Specifically, epidemics on a temporal network can be sub-critical (i.e., the number of infectious nodes exponentially decays over time) over periods even if the SIS dynamics at any point of time is super-critical (i.e., the number of infectious nodes exponentially grows over time if the momentary static network is used forever). In section \ref{sec:par}, we showcase the Parrondo paradox on a two-node periodic switching network. Section \ref{sec:interaction} shows the relationship between the amount of interaction and epidemic spreading. Section \ref{sec:generality_paradox} illustrates the generality of the Parrondo paradox. We show relationships between the Parrondo paradox and anti-phase oscillation in section \ref{sec:decomposition}. Section \ref{sec:perturbation} presents a perturbation theory for the largest Floquet exponent.

\subsection{Parrondo paradox for a periodic switching network with two nodes}\label{sec:par}

We start by studying the SIS dynamics on a minimal periodic switching network with two nodes and $\ell = 2$ static networks to be switched between. We say that the Parrondo paradox occurs if the largest eigenvalues of ${\textbf M}^{(1)} \equiv \beta {\textbf A}^{(1)} - \mu {\textbf I}$ and that of ${\textbf M}^{(2)} \equiv \beta {\textbf A}^{(2)} - \mu {\textbf I}$, i.e.,
$\lambda_{\max}(\textbf{M}^{(1)})$ and $\lambda_{\max}(\textbf{M}^{(2)})$, are both positive and the largest Floquet exponent, $\lambda_\text{F}$, is negative. In this case, the SIS dynamics at any point of time is above the epidemic threshold such that the fraction of infectious nodes exponentially grows over time if either ${\textbf A}^{(1)}$ or ${\textbf A}^{(2)}$ is permanently used, whereas the epidemic exponentially decays over time on the periodic switching network. Generalizing this definition to the case of $\ell > 2$ is straightforward. We only study the case of $\ell=2$ in this article.

\begin{figure}[t]
\centering
\includegraphics[width=1\textwidth]{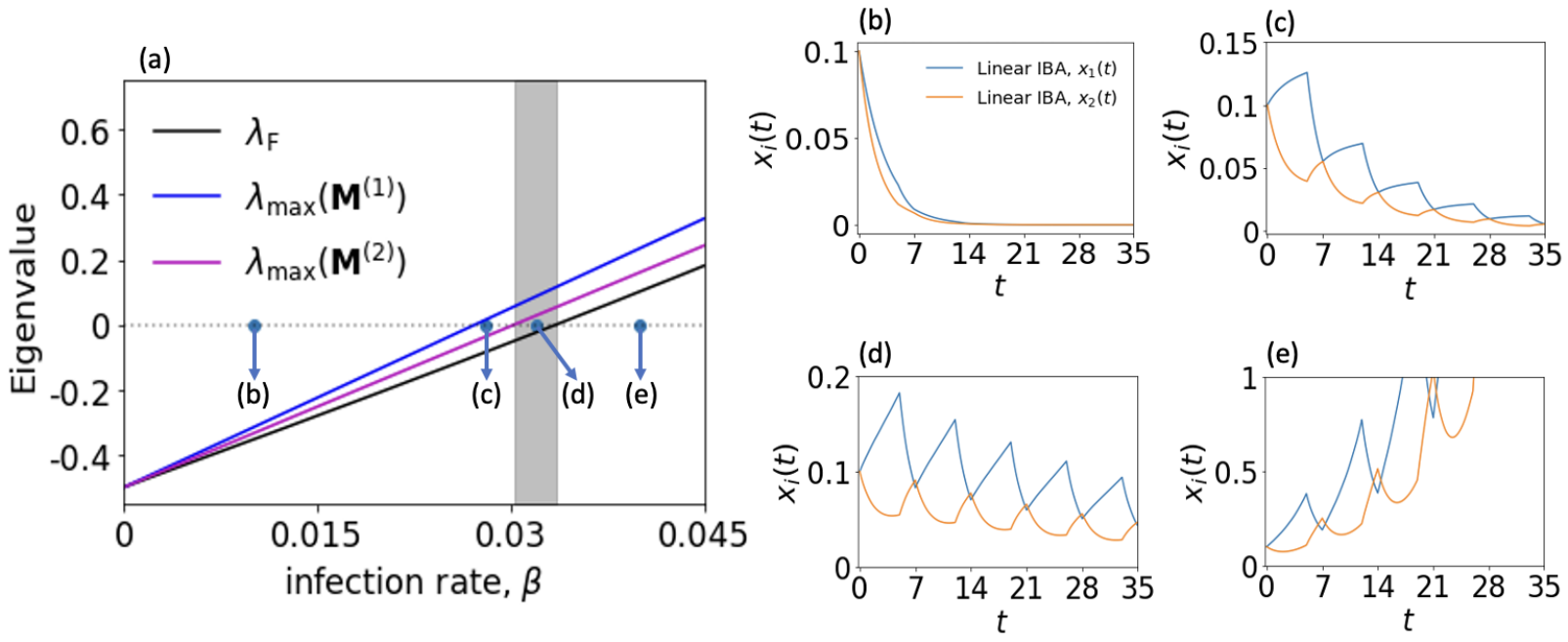}
\caption{Parrondo paradox for the periodic switching network with two nodes shown in Fig.~\ref{fig:examplenetwork}. We set $\mu = 0.5$ and vary $\beta$. (a) Largest Floquet exponent, $\lambda_{\text{F}}$ (black line), largest eigenvalue of $\textbf{M}^{(1)}$ (blue line), and that of $\textbf{M}^{(2)}$ (magenta line) as a function of $\beta$. The shaded region shows the range of $\beta$ in which the Parrondo paradox occurs.
(b)--(e) Time courses of $x_1(t)$ and $x_2(t)$, which represent the infection probabilities in two nodes, for the four values of $\beta$ indicated in (a):
(b) $\beta=0.01$, (c) $\beta=0.028$, (d) $\beta=0.032$, and (e) $\beta=0.04$.
}
\label{Floquet_multipliers_of_network}  
\end{figure}

We show an example of the Parrondo paradox in Fig.~\ref{Floquet_multipliers_of_network}, where we set
$\mu=0.5$. Figure~\ref{Floquet_multipliers_of_network}(a) shows $\lambda_{\text{F}}$, 
$\lambda_{\max}(\textbf{M}^{(1)})$, and $\lambda_{\max}(\textbf{M}^{(2)})$
as a function of $\beta$ for the periodic switching network shown in Fig.~\ref{fig:examplenetwork}, i.e., the one with $\ell=2$, $T=7$, $\tau_1 = 5$, $\textbf{A}^{(1)}=\begin{pmatrix} 17&5\\5&0\end{pmatrix}$, and $\textbf{A}^{(2)}=\begin{pmatrix} 0&5\\5&15\end{pmatrix}$. As expected, $\lambda_{\text{F}}$, $\lambda_{\max}(\textbf{M}^{(1)})$, 
and $\lambda_{\max}(\textbf{M}^{(2)})$ increase as the infection rate, $\beta$, increases. For static networks with adjacency matrices $\textbf{A}^{(1)}$ and $\textbf{A}^{(2)}$, the epidemic threshold, which is the value of $\beta$ at which $\lambda_{\max}(\textbf{M}^{(1)})$ or $\lambda_{\max}(\textbf{M}^{(2)})$ is equal to $0$,
denoted by $\beta_{1}^*$ and $\beta_2^*$, respectively, is $\beta_1^* \approx 0.0272$ and $\beta_2^* \approx 0.0303$, where $\approx$ represents ``approximately equal to''. However, the epidemic threshold for the periodic switching network, denoted by $\beta_{\text{F}}^*$, is larger than both of these values, i.e., $\beta_{\text{F}}^* \approx 0.0335$. 

The shaded region in Fig.~\ref{Floquet_multipliers_of_network}(a) represents the range of $\beta$ for which we have super-critical (i.e., above the epidemic threshold) dynamics on the static networks (i.e., $\lambda_{\max}(\textbf{M}^{(1)}), \lambda_{\max}(\textbf{M}^{(2)})>0$) and sub-critical dynamics on the periodic switching network (i.e., $\lambda_{\text{F}}<0$). In this situation, the Parrondo paradox is occurring such that the disease-free equilibrium is the unique stable equilibrium for Eq.~\eqref{linearize_sismodel_conts_temporal}, whereas, if the network were static then the disease-free equilibrium would be unstable and the number of infectious nodes would exponentially grow over time. Now we further explore the SIS dynamics at the $\beta$ values indicated by the circles in Fig.~\ref{Floquet_multipliers_of_network}(a). In Fig.~\ref{Floquet_multipliers_of_network}(b), (c), (d), and (e), we show $x_1(t)$ and $x_2(t)$ for $\beta = 0.01$, $\beta = 0.028$, $\beta = 0.032$, and $\beta = 0.04$ respectively. 
We selected these four values of $\beta$ because each of them represents a qualitatively different situation regarding whether 
$\lambda_{\text{F}}$, $\lambda_{\max}(\textbf{M}^{(1)})$, or $\lambda_{\max}(\textbf{M}^{(2)})$ is positive or negative.
In Fig.~\ref{Floquet_multipliers_of_network}(b), $x_1(t)$ and $x_2(t)$ exponentially decay at any instantaneous time and in a long term. This result is expected because the growth rates of infection, $\lambda_{\text{F}}$, $\lambda_{\max}(\textbf{M}^{(1)})$, and $\lambda_{\max}(\textbf{M}^{(2)})$, are all negative at this value of $\beta$. 
The time courses of $x_1(t)$ and $x_2(t)$ show anti-phase behavior when $\beta$ is poised near (see Fig.~\ref{Floquet_multipliers_of_network}(c)) or within (see Fig.~\ref{Floquet_multipliers_of_network}(d)) the region in which the Parrondo paradox is present. In Fig.~\ref{Floquet_multipliers_of_network}(e), both $x_1(t)$ and $x_2(t)$ exponentially grow over time in the long term. 
This last result is expected because $\lambda_{\text{F}}$, $\lambda_{\max}(\textbf{M}^{(1)})$, and $\lambda_{\max}(\textbf{M}^{(2)})$ are all positive at this value of $\beta$. 

To show that the Parrondo paradox can also be present in agent-based SIS dynamics, we simulated the stochastic, agent-based SIS model by the direct method of the Gillespie algorithm (e.g., \cite{masuda2022gillespie}). To this end, we generated a network with $N=2000$ nodes and two equally sized communities using a stochastic block model (SBM)~\cite{abbe2017community, lee2019review}. See \ref{appendix:gillespie} for the network generation. We used the entries of $\textbf{A}^{(1)}$ and $\textbf{A}^{(2)}$ to inform the edge probabilities for the periodic switching SBM and ran SIS simulations for the four values of $\beta$ used in Fig.~\ref{Floquet_multipliers_of_network}(b)--(e).
We assumed that there are initially 100 infectious nodes out of the 1000 nodes, which we selected uniformly at random, in each block, or community. The results shown in  \ref{appendix:gillespie} indicate the presence of the Parrondo paradox in these stochastic simulations although the paradox's outcome is less eminent than in the case of ODE simulations presented in this section. This discrepancy is partly due to the nonlinearity of the SIS dynamics, which is also present in the deterministic SIS dynamics, i.e., the IBA before the linearization (see Eq.~\eqref{nonlinear_final_SIS_temporal}).

\subsection{Relationships between the amount of interaction and epidemic spreading \label{sec:interaction}}

We propose that an advantage of the Parrondo paradox, or using periodic switching networks showing it, is to have comparably less epidemic spreading for the given total amount of interaction between individuals. 
To examine this point, we define the total amount of interaction between individuals in the case of static networks by the sum of all the entries of the adjacency matrix. In the case of undirected networks, this definition double counts the off-diagonal entries but provides a correctly normalized total amount of interaction. Imagine a static SBM with two blocks in which each block has $N/2$ individuals and the edge exists between two individuals in the first subpopulation with probability $a_{11}/N$, for example. Then, the expected number of edges, which is equivalent to the amount of interaction in the case of static networks, within the first subpopulation is $\frac{a_{11}}{N} \times \left[ \frac{1}{2} \times \frac{N}{2}\left(\frac{N}{2}-1\right)\right] \approx \frac{N a_{11}}{8}$. The expected number of edges between the two subpopulations is $\frac{a_{12}}{N} \times \left(\frac{N}{2}\right)^2 = \frac{N a_{12}}{4}$, justifying the double counting. For periodic switching networks, we further take the time average over one cycle of the total amount of interaction. The total amount of interaction for periodic switching networks is equal to $\sum_{\ell'=1}^{\ell} \tau_{\ell'}\sum_{i=1}^N \sum_{j=1}^N a_{ij}^{(\ell')} / T$.

For the periodic switching network used in Fig.~\ref{Floquet_multipliers_of_network}, we show in Fig.~\ref{amount_of_interaction} the leading eigenvalues (i.e., $\lambda_{\text{F}}$, 
$\lambda_{\max}(\textbf{M}^{(1)})$, and $\lambda_{\max}(\textbf{M}^{(2)})$) as a function of the total amount of interaction multiplied by the infection rate, $\beta$; this product represents the total infection potential in the population.  Figure~\ref{amount_of_interaction} indicates that the two static networks yield larger leading eigenvalues than the periodic switching network with the same total amount of interaction, supporting our proposal. 

\begin{figure}[t]
\centering
\includegraphics[scale=0.5]{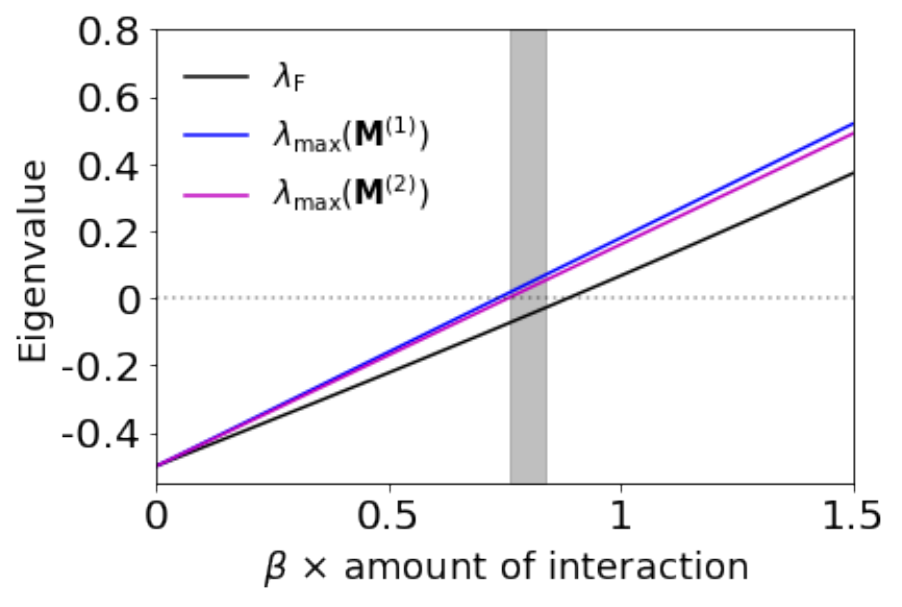}
\caption{Leading eigenvalues as a function of the total amount of interaction for a two-node periodic switching network and the constituent static networks. The value of $\mu$ $(=0.5)$, $\textbf{A}^{(1)}$, $\textbf{A}^{(2)}$, and the switching schedule (i.e., $\ell=2$, $T=7$, and $\tau_1 = 5$, and hence $\tau_2 = T - \tau_1 = 2$) are the same as those used in Fig.~\ref{Floquet_multipliers_of_network}. The shaded region represents the region of ``$\beta \times (\text{amount of interaction})$'' in which the Parrondo paradox occurs. }
\label{amount_of_interaction}  
\end{figure}

\subsection{Generality of the Parrondo paradox \label{sec:generality_paradox}} 

Next we examine the generality of the Parrondo paradox with respect to the connection strength between and within subpopulations, the number of subpopulations up to $N=5$, and larger networks of individuals (as opposed to networks of subpopulations) with distinct two subpopulations. First, we analyze some other two-node periodic switching networks. We use the same switching schedule (i.e., $\ell=2$ alternating static networks, period $T=7$, the duration of the first static network $\tau_1 = 5$, and hence $\tau_2 = T - \tau_1 = 2$) in the remainder of this paper including here. We consider three one-parameter families of two-node periodic switching networks and calculated the epidemic threshold for each periodic switching network. We show in Fig.~\ref{bifuraction_graph}(a)--(c) the epidemic threshold for the periodic switching network (i.e., $\beta_{\text{F}}^*$), first static network (i.e., $\beta_{1}^*$), and second static network (i.e., $\beta_{2}^*$). We set $\textbf{A}^{(1)}=\begin{pmatrix}\gamma&5\\5&0\end{pmatrix}$ and $\textbf{A}^{(2)}=\begin{pmatrix} 0&5\\5&15\end{pmatrix}$ in Fig.~\ref{bifuraction_graph}(a), $\textbf{A}^{(1)}=\begin{pmatrix}17&5\\5&\gamma\end{pmatrix}$ and $\textbf{A}^{(2)}=\begin{pmatrix}\gamma&5\\5&15\end{pmatrix}$ in Fig.~\ref{bifuraction_graph}(b), and $\textbf{A}^{(1)}\begin{pmatrix} 17&\gamma\\\gamma&0\end{pmatrix}$ and $\textbf{A}^{(2)}=\begin{pmatrix} 0&\gamma\\\gamma&15\end{pmatrix}$ in Fig.~\ref{bifuraction_graph}(c). The shaded regions in the figure are the ranges of $\gamma$ and $\beta$ for which the Parrondo paradox occurs, i.e., where $\lambda_{\text{F}}<0$, $\lambda_{\max}(\textbf{M}^{(1)} )> 0$, and $\lambda_{\max}(\textbf{M}^{(2)})>0$ are simultaneously satisfied. We find that the paradox occurs for some ranges of $\gamma$.

To verify that the Parrondo paradox can also be present in large networks, we similarly calculated the leading eigenvalues for networks with $N=2000$ nodes generated by an SBM that creates large networks with community structure as discussed in \ref{appendix:gillespie}. We define the SBMs to have two communities, each with $N/2 = 1000$ nodes, and obtain their edge probabilities using the entries in the adjacency matrices of the three families of $2\times 2$ periodic switching networks used in Fig.~\ref{bifuraction_graph}(a)--(c). Summing the entries across SBM blocks recovers the original $2\times 2$ matrices, in expectation.
We show in Fig.~\ref{bifuraction_graph}(d)--(f) the epidemic thresholds and the region of pairs of $\gamma$ and $\beta$ in which the Parrondo paradox is present for the networks with $N=2000$ nodes generated by the SBM. The results for the SBM are close to those for the corresponding two-node networks shown in Fig.~\ref{bifuraction_graph}(a)--(c).

\begin{figure}[t]
\begin{center}
\includegraphics[width=1.0\linewidth]{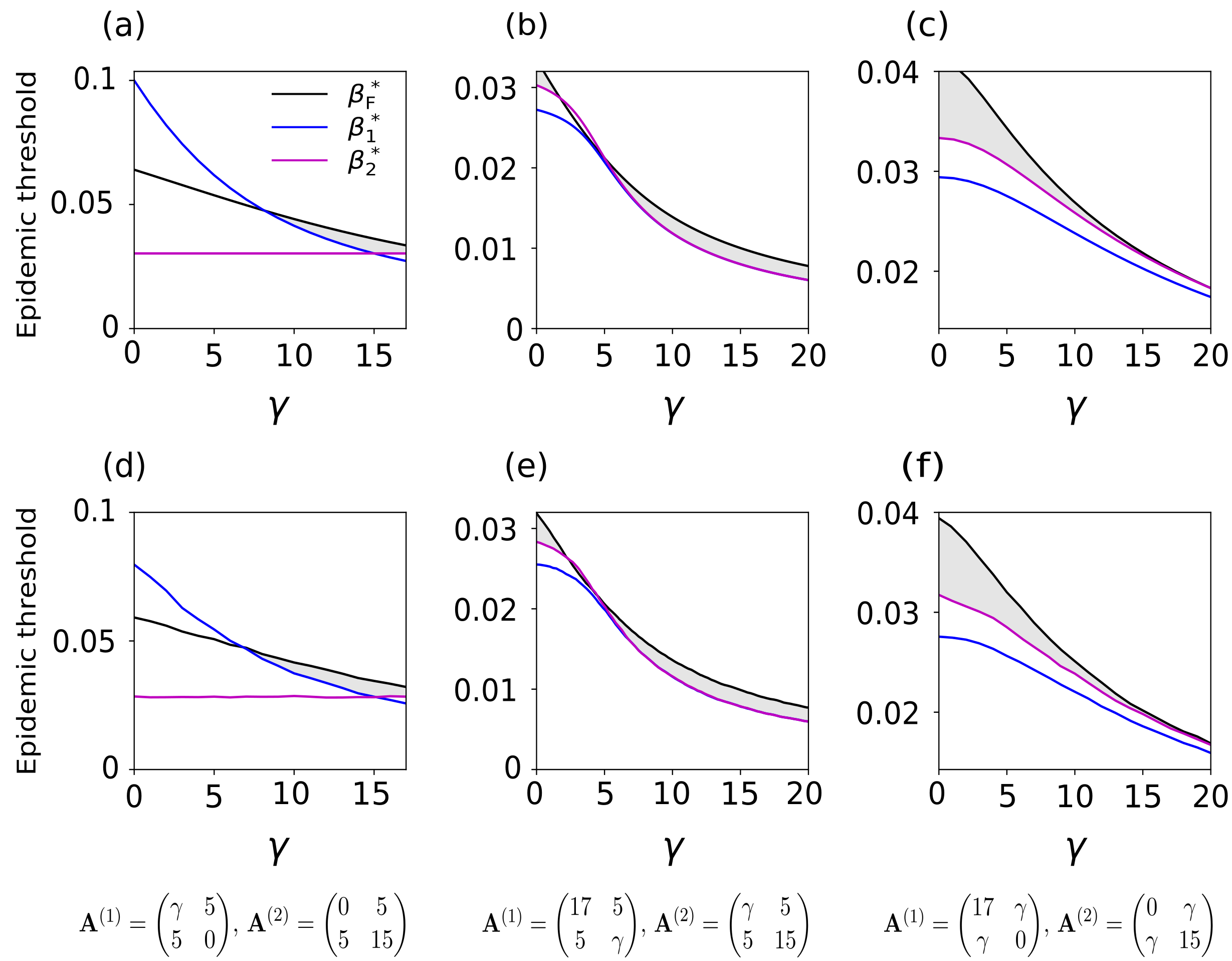}
\caption{Parrondo paradox over other periodically switching networks with two nodes. We show the epidemic threshold for the periodic switching network (i.e., $\beta_{\text{F}}^*$; black lines),
 $\textbf{A}^{(1)}$ (i.e., $\beta_{1}^*$; blue lines), and $\textbf{A}^{(2)}$ (i.e., $\beta_{2}^*$; magenta lines). The shaded region shows the values of $\gamma$ and $\beta$ for which the paradox occurs. (a)--(c): Two-node networks. (d)--(f): Networks with $N=2000$ nodes and two communities generated by the SBM. Panels (a) and (d) use the same adjacency matrices $\textbf{A}^{(1)}$ and $\textbf{A}^{(2)}$, indicated below (d). Same for (b) and (e), and for (c) and (f).} 
 \label{bifuraction_graph}
\end{center}
\end{figure}

Next we conduct an experiment with random networks to study the fraction of periodic switching networks for which the paradox occurs. To this end, we assume that the adjacency matrices $\textbf{A}^{(1)}$ and $\textbf{A}^{(2)}$ are symmetric and parameterize them as $\textbf{A}^{(1)}=\begin{pmatrix}\gamma_1&\gamma_2\\\gamma_2&\gamma_3\end{pmatrix}$ and $\textbf{A}^{(2)}=\begin{pmatrix} \gamma_4&\gamma_5\\\gamma_5&\gamma_6\end{pmatrix}$. Without loss of generality, we set $\gamma_1=1$; multiplying $\gamma_1$, $\ldots$, $\gamma_6$ by a factor of $c$ ($>0$) is equivalent to using the original values of $\gamma_1$, $\ldots$, $\gamma_6$ and changing $\beta$ to $c\beta$. Then, we independently draw $\gamma_2$, $\gamma_3$, $\gamma_4$, $\gamma_5$, and $\gamma_6$ from the uniform density on $[0, 5]$. For each pair of $\textbf{A}^{(1)}$ and  $\textbf{A}^{(2)}$, we computed $\beta_{\text{F}}^*$, $\beta_1^*$, and $\beta_2^*$  as the values of $\beta$ at which $\lambda_{\text{F}}$, $\lambda_{\max}(\textbf{M}^{(1)})$, or $\lambda_{\max}(\textbf{M}^{(2)})$ is equal to $0$, respectively.

We defined the Parrondo paradox as a situation in which $\lambda_{\text{F}} < 0$, $\lambda_{\max}(\textbf{M}^{(1)}) > 0$, and $\lambda_{\max}(\textbf{M}^{(2)}) > 0$ simultaneously hold true. Because these leading eigenvalues monotonically increase as $\beta$ increases, 
a periodic switching network shows the Parrondo paradox for some values of $\beta$ if and only if
$\beta_{\text{F}}^* > \max\{ \beta_1^*, \beta_2^* \}$, as shown by the shaded regions in Figs.~\ref{Floquet_multipliers_of_network} and \ref{bifuraction_graph}. We generated $10^4$ independent pairs of $\textbf{A}^{(1)}$ and $\textbf{A}^{(2)}$ to find that the paradox is present in 
28.7\% of the periodic switching networks. Therefore, we conclude that the Parrondo paradox is not a rare phenomenon, at least for $2\times 2$ adjacency matrices.  We did not find any periodic switching network that showed the opposite paradoxical behavior, i.e.,  $\beta_{\text{F}}^* < \min \{ \beta_1^*, \beta_2^* \}$.
 
As our final set of experiments to investigate the generality of the Parrondo paradox, we study how it is affected by the number of node. We first examine the possibility of the paradox with $3\times 3$ periodic switching networks. Let us consider arbitrarily selected one-parameter families of three-node periodic switching networks. The epidemic thresholds for the periodic switching network (i.e., $\beta_{\text{F}}^*$) and those for the static networks (i.e., $\beta_{1}^*$, $\beta_{2}^*$) are shown in Fig.~\ref{bifuraction_graph_3_by_3}(a) for $\textbf{A}^{(1)}=\begin{pmatrix}\gamma&5&5\\5&0&5\\ 5&5&15\end{pmatrix}$ and $\textbf{A}^{(2)}=\begin{pmatrix}14&5&5\\5&0&5\\ 5&5&\gamma\end{pmatrix}$ and Fig.~\ref{bifuraction_graph_3_by_3}(b) for $\textbf{A}^{(1)}=\begin{pmatrix}24&\gamma&5\\5&0&\gamma\\ \gamma&5&15\end{pmatrix}$ and $\textbf{A}^{(2)}=\begin{pmatrix}14&\gamma&5\\ 5&0&\gamma\\ \gamma&5&24\end{pmatrix}$. The shaded regions represent the region of $\gamma$ and $\beta$ values in which the Parrondo paradox occurs. We observe that the paradox occurs for a range of $\gamma$ in both families of periodic switching networks. For the periodic switching network shown in Fig.~\ref{bifuraction_graph_3_by_3}(a), the ranges of $\gamma$ for which the paradox occurs are disjoint; the paradox occurs for small and large $\gamma$ values, but not for intermediate $\gamma$ values.

\begin{figure}[t]
\begin{center}
\includegraphics[scale=0.4]{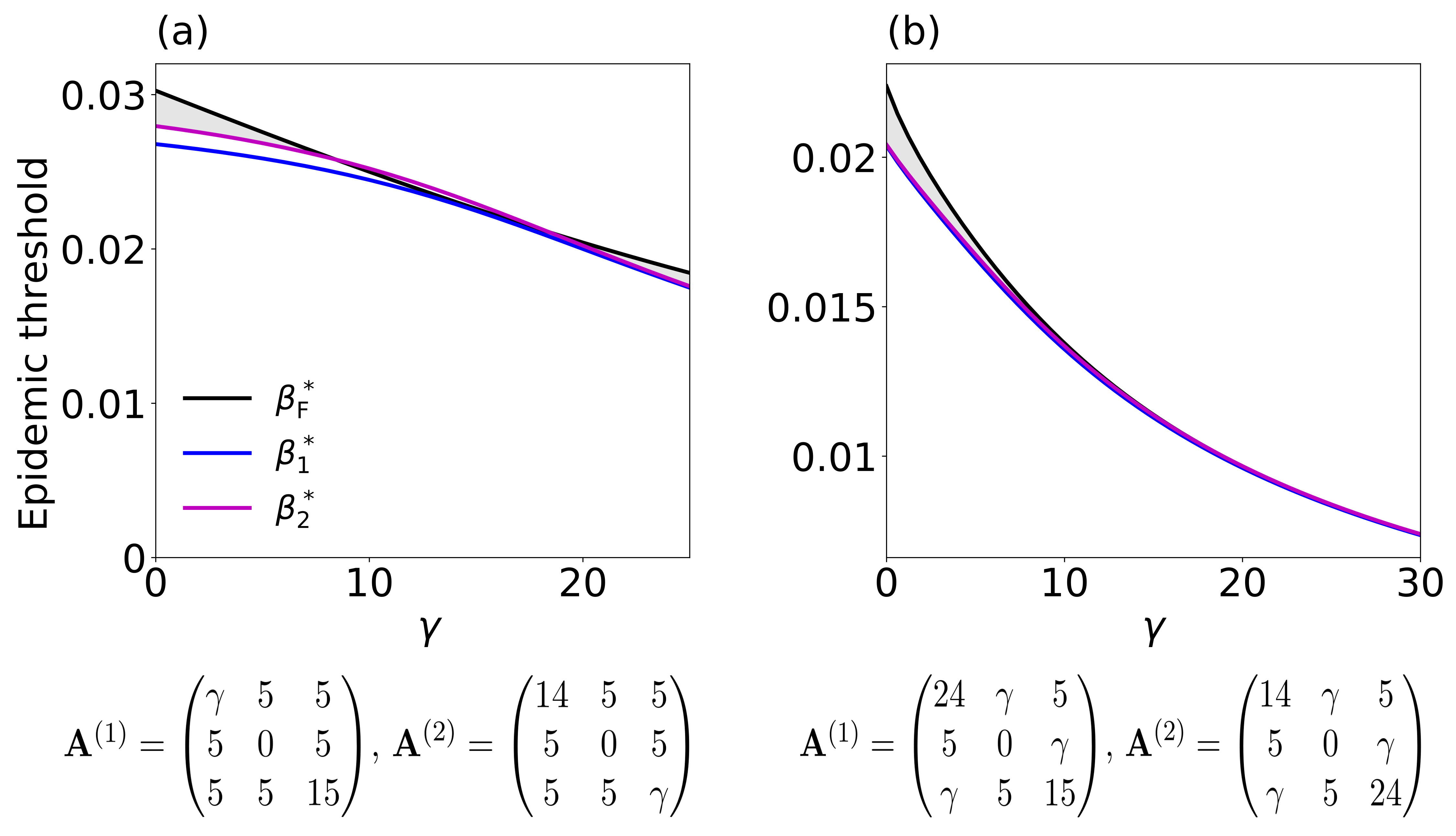}
\caption{Epidemic thresholds for periodic switching networks and static networks with three nodes. 
See the caption of Fig.~\ref{bifuraction_graph} for the legends.
The shaded regions represent the values of $\gamma$ and $\beta$ for which the Parrondo paradox occurs.
In (b), the blue line is not visible because it is almost completely hidden behind the magenta line.}
 \label{bifuraction_graph_3_by_3}
\end{center}
\end{figure}

To further assess the generality of the Parrondo paradox for periodic switching networks with three nodes,
we retain the same switching schedule for simplicity and generate symmetric adjacency matrices parameterized as $\textbf{A}^{(1)}=\begin{pmatrix}\gamma_1&\gamma_2&\gamma_3\\\gamma_2&\gamma_4&\gamma_5\\\gamma_3&\gamma_5&\gamma_6\end{pmatrix}$ and $\textbf{A}^{(2)}=\begin{pmatrix} \gamma_7&\gamma_8&\gamma_9\\\gamma_8&\gamma_{10}&\gamma_{11}\\\gamma_9&\gamma_{11}&\gamma_{12}\end{pmatrix}$. We set $\gamma_1=1$ without loss of generality. Then, we independently draw $\gamma_2$, $\ldots$, $\gamma_{12}$ from the uniform density on $[0, 5]$. We find that, among $10^4$ pairs of the generated $\textbf{A}^{(1)}$ and $\textbf{A}^{(2)}$, the paradox is present, with $\beta_{\text{F}}^* > \max\{ \beta_1^*, \beta_2^* \}$, in 12.0\% of the periodic switching networks, and that there is no case with $\beta_{\text{F}}^* < \min \{ \beta_1^*, \beta_2^* \}$.
 
We similarly analyzed randomly generated periodic switching networks with four and five nodes. We found that among $10^4$ pairs of randomly generated $\textbf{A}^{(1)}$ and $\textbf{A}^{(2)}$, the paradox is present in $5.21\%$ and $1.20\%$ of the periodic switching networks with four and five nodes, respectively. There was no periodic switching network showing $\beta_{\text{F}}^* < \min \{ \beta_1^*, \beta_2^* \}$. These results are qualitatively the same as those for periodic switching networks with fewer (i.e., two or three) nodes.
 
\subsection{Relationships between the Parrondo paradox and anti-phase oscillations\label{sec:decomposition}}

To better understand mechanisms of the Parrondo paradox, we carried out additional experiments motivated by the following observation.
Figure~\ref{Floquet_multipliers_of_network} suggests that the Parrondo paradox may be characterized by an anti-phase oscillation between $x_1(t)$ and $x_2(t)$. However, this association is not straightforward already with Fig.~\ref{Floquet_multipliers_of_network} because the figure shows anti-phase oscillations for a $\beta$ value at which the paradox occurs (see Fig.~\ref{Floquet_multipliers_of_network}(d)) as well as for a $\beta$ value at which it does not (see Fig.~\ref{Floquet_multipliers_of_network}(c)). Nevertheless, the anti-phase oscillation is absent for $\beta$ values that are far from where the paradox is present (see Fig.~\ref{Floquet_multipliers_of_network}(b) and \ref{Floquet_multipliers_of_network}(e)).

We measure the extent of anti-phase oscillation by the fraction of time, $q$, during one period in which $\text{d}x_1(t)/\text{d}t$ and $\text{d}x_2(t)/\text{d}t$ have the opposite sign. We measure $q$ in the fifth cycle, i.e., using $t \in [28, 35]$, and use the initial condition $(x_1(0), x_2(0)) = (0.1, 0.1)$. Consider the $10^4$ periodic switching networks with two communities used in section \ref{sec:generality_paradox}. We recall that $28.7\%$ of the periodic switching networks (i.e., 2,870 networks) show the Parrondo paradox. For each of these 2,870 networks, we select the middle value of the range of $\beta$ in which the paradox occurs and compute $q$ at the selected $\beta$ value. 
At this value of $\beta$, perfect anti-phase synchronization, i.e., $q=1$, occurs for 2,867 ($99.9\%$) of the 2,870 networks. For comparison, we have also computed $q$ for the rest of the networks (i.e., 7,130 out of the $10^4$ networks). We find that the maximum, minimum, mean and standard deviation of $q$ for these 7,130 networks is $0.974$, $0$, $0.366$, and $0.305$, respectively. 
Therefore, although some networks show strong anti-phase oscillation despite the lack of the Parrondo paradox behavior, most networks not showing the paradox have substantially smaller $q$ values than the networks showing the paradox.

We show examples of anti-phase dynamics in \ref{appendix:antiphase_example}.
It should be noted that $q=1$ holds true for both inside and outside the range of $\beta$ in which the paradox is present for some networks. In our examples, $q$ is much larger for the periodic switching networks showing the paradox than those not showing the paradox. We also observe some anti-phase behavior in the four periodic switching networks that do not show the Parrondo paradox for any value of $\beta$, but to a limited extent (i.e., $q \leq 0.7$).
These results are consistent with the population results with $10^4$ networks described above.
In sum, we conclude that anti-phase oscillation occurs more frequently in the presence or approximate presence of the Parrondo paradox than otherwise.

\subsection{Perturbation theory for the largest Floquet exponent} \label{sec:perturbation} 

Regardless of the presence or absence of the Parrondo paradox as one varies $\beta$, Fig.~\ref{Floquet_multipliers_of_network} suggests that the relationship between the largest Floquet exponent (i.e., $\lambda_{\text{F}}$) and $\beta$ is apparently close to linear for an extended range of $\beta$. Therefore, we derive the first-order approximation of $\lambda_{\text{F}}$ in terms of $\beta$ and assess its accuracy in this section.
\begin{thm}\label{perurbation_of_eigenvalue}(Perturbation of simple eigenvalues~\cite{atkinson1991introduction}).
Let  $\hat{\textbf{M}}$ be a symmetric $N\times N$ matrix with the left eigenvector $\{\textbf{u}^{(i)}\}$ and right eigenvector $\{\textbf{v}^{(i)}\}$ associated with eigenvalue $\lambda_i$, which is assumed to be simple (with $i \in \{1, \ldots, N\}$). Without loss of generality, the Euclidean norm of each eigenvector is equal to $1$. Consider a perturbation of $\hat{\textbf{M}}$ given by $\hat{\textbf{M}}(\beta)=\hat{\textbf{M}}+\beta \Delta \hat{\textbf{M}}$. The eigenvalues of $\hat{\textbf{M}}(\beta)$ when $\left| \beta \right| \ll 1$ satisfy
\begin{equation}
\lambda_i(\beta)=\lambda_i+\beta \lambda_i'(0)+\mathcal{O}(\beta^2),
\label{perturb_eigenvalue}
\end{equation}
where
\begin{equation}
\lambda_i'(0)= \frac{[\textbf{u}^{(i)}]^{\top} \Delta\hat{\textbf{M}}\textbf{v}^{(i)}}{[\textbf{u}^{(i)}]^{\top} \textbf{v}^{(i)}}.
\label{perturb_eigenvalue_explaination}
\end{equation}
\end{thm}

\begin{lemma}[First-order approximation of matrix $\mathcal{T}$] The first-order perturbation of matrix $\mathcal{T}$ (defined in Eq.~\eqref{linearize_sismodel_conts_temporal_new}) is given by 
\begin{equation}
\begin{split}
\mathcal{T}(\beta)&= \mathcal{T}(0)+\beta \Delta \mathcal{T}(0)+\mathcal{O}(\beta^2)\\
 &=\left( \textbf{I}+\beta T \overline{\textbf{A}} \right) e^{-T\mu \textbf{I}}+\mathcal{O}(\beta^2).
\end{split}
\label{eq:description4.1}
\end{equation}
where
\begin{equation}
\overline{\textbf{A}} \equiv \frac{1}{T} \sum_{\ell'=1}^{\ell} \tau_{\ell'} \textbf{A}^{(\ell')}
\end{equation}
is the time-averaged adjacency matrix.
\end{lemma}
\begin{proof} 
Using Eq.~\eqref{linearize_sismodel_conts_temporal_new}, we obtain
\begin{equation}
\begin{split}
\mathcal{T}(\beta)& =
%
%
%
%
e^{\beta \tau_{\ell} \textbf{A}^{(\ell)}}e^{\beta \tau_{\ell-1} \textbf{A}^{(\ell-1)}}\cdots e^{\beta \tau_{1}  \textbf{A}^{(1)}}e^{- \mu T \textbf{I}},
\end{split}
\label{eq:description4.2}
\end{equation}
which leads to
\begin{equation}
\mathcal{T}(0) = e^{-\mu T \textbf{I}}
\label{eq:description4.3}
\end{equation}
and
\begin{equation}
\begin{split}
\Delta \mathcal{T}(0) & =\left.\frac{\text{d}\mathcal{T}(\beta)}{\text{d}\beta}\right|_{\beta=0}\\
&= \tau_{\ell} \textbf{A}^{(\ell)} e^{\beta \tau_{\ell} \textbf{A}^{(\ell)}}e^{\beta \tau_{\ell-1} \textbf{A}^{(\ell-1)}}\cdots e^{\beta \tau_{1}  \textbf{A}^{(1)}}e^{- \mu T \textbf{I}} \\
& +e^{\beta \tau_{\ell} \textbf{A}^{(\ell)}} \tau_{\ell-1} \textbf{A}^{(\ell-1)}e^{\beta \tau_{\ell-1} \textbf{A}^{(\ell-1)}}\cdots e^{\beta \tau_{1}  \textbf{A}^{(1)}} e^{- \mu T \textbf{I}}+\cdots\\
&+e^{\beta \tau_{\ell} \textbf{A}^{(\ell)}} e^{\beta \tau_{\ell-1} \textbf{A}^{(\ell-1)}}\cdots \tau_1 \textbf{A}^{(1)} e^{\beta \tau_{1}  \textbf{A}^{(1)}}e^{- \mu T \textbf{I}}
\Bigg. \Bigg|_{\beta=0}\\
&= \left( \tau_{\ell} \textbf{A}^{(\ell)} + \tau_{\ell-1} \textbf{A}^{(\ell-1)}+\cdots+\tau_1 \textbf{A}^{(1)}\right)e^{-\mu T \textbf{I}} \\
&= T e^{- \mu T \textbf{I}} \overline{\textbf{A}}.
\end{split}
\label{eq:description4.5}
\end{equation}
Equations \eqref{eq:description4.3} and \eqref{eq:description4.5} imply \eqref{eq:description4.1}.
\end{proof}
\begin{lemma}[First-order approximation of the largest eigenvalue of $\mathcal{T}$] For $\left| \beta \right| \ll 1$, we obtain
\begin{equation}
\lambda_{\max}(\mathcal{T})(\beta) = e^{-T\mu}+ \beta \lambda_{\max}(\overline{\textbf{A}}) T e^{-T\mu} + \mathcal{O}(\beta^2).
\label{eq:description4.7}
\end{equation}
\end{lemma}
\begin{proof} Equation~\eqref{eq:description4.3} implies that
\begin{equation}
\lambda_{\max}(\mathcal{T})(0)=e^{-T\mu}.
\label{eq:description4.10}
\end{equation}
Let $\textbf{u}$ and $\textbf{v}$ be the dominant left and right eigenvectors of $\overline{\textbf{A}}$, respectively. By combining Eqs.~\eqref{perturb_eigenvalue_explaination} and \eqref{eq:description4.5}, we obtain
\begin{equation}
\begin{split}
\lambda'_{\max}(\mathcal{T})(0)
&= \frac{\textbf{u}^{\top} (T e^{-T\mu} \overline{\textbf{A}}) \textbf{v} }{\textbf{u}^{\top}\textbf{v}}\\
&= \lambda_{\max}(\overline{\textbf{A}}) T e^{-T\mu}.
\end{split}
\label{eq:description4.11}
\end{equation}
By substituting Eqs.~\eqref{eq:description4.10} and \eqref{eq:description4.11} in Eq.~\eqref{perturb_eigenvalue},
we obtain Eq.~\eqref{eq:description4.7}.
\end{proof}

\begin{lemma}[First-order approximation of $\lambda_{\text{F}}$] For $\left| \beta \right| \ll 1$, we obtain
\begin{equation}
\lambda_{\text{F}}(\beta)=\lambda_{\text{F}}(0)+\beta \lambda'_{\text{F}}(0)+\mathcal{O}(\beta^2)=-\mu+\beta \lambda_{\max}(\overline{\textbf{A}})+\mathcal{O}(\beta^2).   
\label{eq:lambda_F-expand}
\end{equation}
\end{lemma}
\begin{proof} By definition, we obtain $\lambda_{\text{F}}(\beta)=\frac{\ln \lambda_{\max}(\mathcal{T})(\beta)}{T}$. Using this, Eq.~\eqref{eq:description4.10}, and Eq.~\eqref{eq:description4.11}, we obtain
\begin{equation}
\lambda_{\text{F}}(0)=\frac{\ln \lambda_{\max}(\mathcal{T})(0)}{T} = -\mu
\label{eq:lambda_F(0)}
\end{equation}
and 
\begin{equation}
\lambda'_{\text{F}}(0)=\left.\frac{d\lambda_{\text{F}}}{d\beta}\right|_{\beta=0}
= \left.\frac{\partial \lambda_{\text{F }}}{\partial \lambda_{\max}(\mathcal{T})}\frac{\partial \lambda_{\max}(\mathcal{T})}{\partial \beta}\right|_{\beta=0}
= \frac{1}{T \lambda_{\max}(\mathcal{T})(0)} \cdot  \lambda_{\max}(\overline{\textbf{A}}) T e^{-T\mu}
= \lambda_{\max}(\overline{\textbf{A}}).
\label{eq:description4.15}
\end{equation}
Equations~\eqref{eq:lambda_F(0)} and \eqref{eq:description4.15} imply Eq.~\eqref{eq:lambda_F-expand}.
\end{proof}
Equation~\eqref{eq:lambda_F-expand} implies the following corollary.
\begin{corollary}
The epidemic threshold for the periodic switching network under the first-order approximation is given by
\begin{equation}
\beta_{\text{F}}^* = \frac{\mu}{\lambda_{\max}(\overline{\textbf{A}})}.
\label{eq:lambda_F-corollary}
\end{equation}
\end{corollary}

Now we validate our first-order approximation for $2\times 2$ periodic switching networks with $\ell=2$. We 
compare $\lambda_{\text{F}}$ between the exact value and linear approximation in Fig.~\ref{perturbation} for eight arbitrarily selected two-node periodic switching networks.
Figures~\ref{perturbation}(a)--(d) are for four networks showing the Parrondo paradox. We find that the first-order approximation is sufficiently accurate for a range of $\beta$ including the range in which the paradox occurs. The first-order approximation is also accurate for four periodic switching networks in which the paradox is absent in the entire range of $\beta$ (see Fig.~\ref{perturbation}(e)--(h)). Figure~\ref{perturbation} suggests that the first-order approximation tends to be more accurate when the paradox is absent (see Fig.~\ref{perturbation}(e)--(h)) than present (see Fig.~\ref{perturbation}(a)--(d)).

\begin{figure}[t]
\begin{center}
\includegraphics[width=\linewidth]{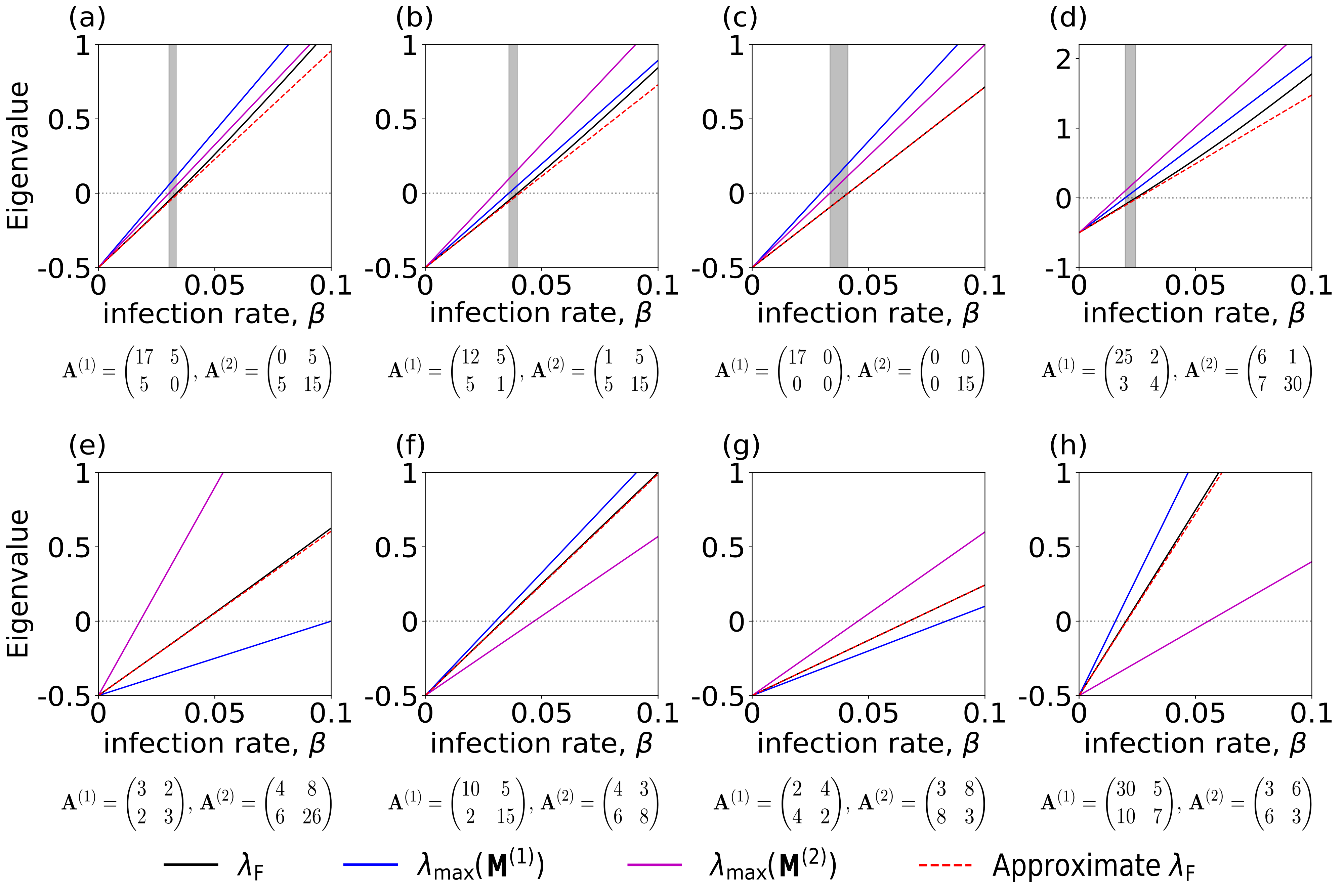}
\caption{
First-order approximation to the largest Floquet exponent for eight two-node periodic switching networks with $k=2$, $T=7$, and $\tau_1 = 5$. We set $\mu=0.5$.
We show the largest Floquet exponent (i.e., $\lambda_ F$) in black, largest eigenvalue of $\textbf{M}^{(1)}$ in blue, that of $\textbf{M}^{(2)}$ in magenta, and the first-order approximation of $\lambda_ F$ in red.
}
 \label{perturbation}
\end{center}
\end{figure}

As an application of the first-order approximation of $\lambda_{\text{F}}$, we examine the fraction of randomly generated periodic switching networks for which the Parrondo paradox occurs, as we did in section~\ref{sec:generality_paradox}. 
Note that the epidemic threshold is given in the form of $\mu/\lambda_{\max}$ 
for the periodic switching network under the first-order approximation as well as for the static networks. Therefore, if the first-order approximation is accurate, the Parrondo paradox happens if 
$\lambda_{\max}(\overline{\textbf{A}})$ is larger than both
$\lambda_{\max}(\textbf{A}^{(1)})$ and $\lambda_{\max}(\textbf{A}^{(2)})$ or smaller than both of them.
In the same manner as that in section~\ref{sec:generality_paradox},
we generated $10^4$ pairs of $\textbf{A}^{(1)}$ and $\textbf{A}^{(2)}$ and computed $\lambda_{\max}(\overline{\textbf{A}})$, $\lambda_{\max}(\textbf{A}^{(1)})$, and $\lambda_{\max}(\textbf{A}^{(2)})$.
We found that $\lambda_{\max}(\overline{\textbf{A}})$ is smaller than both
$\lambda_{\max}(\textbf{A}^{(1)})$ and $\lambda_{\max}(\textbf{A}^{(2)})$ for $31.7\%$ pairs of 
$\textbf{A}^{(1)}$ and $\textbf{A}^{(2)}$.
There was no case in which $\lambda_{\max}(\overline{\textbf{A}})$ is larger than $\lambda_{\max}(\textbf{A}^{(1)})$ and $\lambda_{\max}(\textbf{A}^{(2)})$. These results are similar to those obtained for
$\lambda_{\text{F}}$ (see section~\ref{sec:generality_paradox}).
 
\section{Discussion}\label{sec:discuss}

We showed a Parrondo paradox in the SIS dynamics on periodically switching temporal networks in which an
epidemic decays over a long time while the dynamics is super-critical (i.e., above the epidemic threshold) for the momentary static network at any point of time. We primarily showed the paradox in simple cases involving a periodic alternation of two $2\times 2$ adjacency matrices. Larger networks (with $N=2000$ nodes) with two subpopulations generated by the SBM also showed the paradox, even for a single run of stochastic SIS dynamics.  We further verified that the paradox occurred with adjacency matrices of larger sizes, up to five nodes. However, we found that the fraction of random networks that show the Parrondo paradox decreased as the number of nodes increases. This is an important limitation of the present study. Further investigating the generality of the Parrondo paradox in terms of the number of nodes or subpopulations, the number of static networks constituting a periodic switching network, different epidemic process models, and different network models such as the metapopulation model and hypergraphs, warrants future work.

We pointed out the relationship between the growth rate of the epidemics (including the epidemic threshold, at which the growth rate is equal to $0$) and the Floquet exponent. Leveraging Floquet theory to characterize and intervene into epidemic processes on periodic temporal networks may be fruitful. For example, 
the periodic matrix obtained by the Floquet decomposition (i.e., matrix $\textbf{P}(t)$ in \ref{sec:flo})
may help us to quantify the amplitude of antiphase oscillation between $x_i(t)$ and $x_j(t)$ (with $j \neq i$) or change it through manipulation of $\textbf{P}(t)$.
 
In the context of COVID-19, effects on mitigating spread of infection were numerically simulated on physical proximity and face-to-face contact data and compared between ``rotating" and ``on-off" strategies with two-day or biweekly periodic schedules \cite{mauras2021mitigating}. The rotating strategy by definition imposes that half the population is allowed to go out in odd days or odd weeks, and the other half is allowed to go out in even days or even weeks. In contrast, the on-off strategy allows the entire population to go for work in one day (or one week) and not on the next day (or week), and such a two-day or biweekly pattern is repeated. As expected, any of these four strategies helped towards decreasing the basic reproduction number below $1$. 
Other studies also investigated effects of different on-off strategies on suppressing epidemics
\cite{karincyclic, kaygusuz2021effect}. Additionally, it was shown in Ref.~\cite{mauras2021mitigating} that
among the four strategies, the best and worst strategies were the rotating strategy with the biweekly cycle and the on-off strategy with two-day cycle, respectively. Superiority of a rotating strategy to an on-off strategy was also shown in a different study~\cite{meidan2021alternating}. These comparative results suggest that different quarantine or curfew strategies based on periodic switching networks are feasible and that their efficiency depends on the periodic switching network. Investigating the possibility of the Parrondo paradox and its possible roles in suppressing population-wide infection in these and other practical disease spreading models warrants future study. In particular, both these previous studies and the present work suggest that eliciting anti-phase oscillations between different communities might be a key to suppressing infection while keeping the overall activity level of individuals relatively high.

 \appendix
\renewcommand{\thesection}{{Appendix }\Alph{section}}

\section {Proof that all the entries of $e^{\beta \textbf{E}(T)}$ are positive}\label{lemma_proof}

Here, we show that each entry of the matrix exponential of the time-aggregated adjacency matrix $ \textbf{E}(T)=\int_0^T \textbf{A}(t)\text{d}t$ is positive under the assumption that $\textbf{E}(T)$ is nonnegative and irreducible (i.e., which corresponds to when it is associated with a strongly connected graph with positive edge weights).
Recall that the matrix exponential is defined by $e^{\beta \textbf{E}(T)}=\sum_{k=0}^\infty \frac{1}{k!} [\beta\textbf{E}(T)]^k$, which has entries
\begin{equation}
[e^{\beta \textbf{E}(T)}]_{ij}= \sum_{k=0}^\infty \frac{1}{k!}   [\beta\textbf{E}(T)]_{ij}^k
= \sum_{k=0}^\infty \frac{1}{k!}   \beta^k[\textbf{E}(T)^k]_{ij}.
\label{matrix_exponential_equation}
\end{equation}
Because $\textbf{E}(T)$ is nonnegative and irreducible, for any $(i, j) \in \{1, \ldots, N\}\times \{1, \ldots, N\}$, there exists a $k$ such that   $\frac{1}{k!}  \beta^k[{\bf E}(T)^k ]_{ij}>0$. Specifically, $\frac{1}{k!}  \beta^k[{\bf E}(T)^k ]_{ij}$ is strictly positive if there is at least one path from $i$ to $j$ having length $k$. Notably, $[{\bf E}(T)^k ]_{ij}$ equals the product of edge weights along such a path, which is summed across all possible paths. Let $\overline{k}$ denote such a $k$ value. Then, we note that $\frac{1}{k!} \beta^k [{\bf E} (T)^k ]_{ij}$ must be nonnegative for any $k \neq \overline{k}$, since ${\bf E}(T)$ has nonnegative entries and the product and summation of nonnegative numbers remain nonnegative. Thus, for any $i,j$, the right-hand side of Eq.~\eqref{matrix_exponential_equation} is a positive number (i.e., the term with $k=\overline{k}$) plus the sum of nonnegative numbers (i.e., the terms with $k \neq \overline{k}$), which is a positive number.


 \section{Floquet theory}\label{sec:flo}
%
\renewcommand{\thedefinition}{\Alph{section}.\arabic{definition}}
\renewcommand{\thethm}{\Alph{section}.\arabic{thm}}
\renewcommand{\theremark}{\Alph{section}.\arabic{remark}}

In this section, we give a brief review of Floquet theory for periodic linear dynamical systems in continuous time.

\begin{definition}[Periodic system] Consider a non-autonomous continuous-time dynamical system given by
\begin{equation}
x'(t)=M(t,x(t)).
\label{periodic_equation}
\end{equation}
Dynamical system~\eqref{periodic_equation} is called a $T$-periodic system with period $T$ $(>0)$ if $M(t+T,x)=M(t,x)$ $\forall t$.
\end{definition}

We present Floquet theory for linear $T$-periodic dynamical systems in continuous time given by
\begin{equation}
   \frac{\text{d}\textbf{x}(t)}{\text{d}t}=\textbf{M}(t)\textbf{x}(t),
\label{linear_ode} 
\end{equation}
where $\textbf{M}(t)\in \mathcal{R}^{N\times N}$  is a $T$-periodic matrix, and $\textbf{x}(t)\in \mathcal{R}^{N}$ represents the system's state. Let $\textbf{X}(t)\in \mathcal{R}^{N\times N}$ be the fundamental matrix of system~\eqref{linear_ode}. By definition of the fundamental matrix, the columns of the fundamental matrix are linearly independent solutions of~Eq.~\eqref{linear_ode}.

\begin{thm}\label{floquet_cont}[Floquet's theorem for continuous-time linear periodic systems \cite{floquet1883equations}] 
For system~\eqref{linear_ode} that has a fundamental matrix, there exists a nonsingular $T$-periodic matrix,  $\textbf{P}(t)\in \mathbb{R}^{N\times N}$, and a constant matrix,  $\textbf{B}\in \mathbb{C}^{N\times N}$, such that 
\begin{equation}
   \textbf{X}(t)=\textbf{P}(t)e^{\textbf{B}t}.
\label{floquet_decomposition} 
\end{equation}
Equation~\eqref{floquet_decomposition} is known as the Floquet normal form for $\textbf{X}(t)$.  

\end{thm}
\begin{remark}
 Floquet theory guarantees that there exists a nonsingular constant matrix $\tilde{\textbf{M}} \in \mathcal{R}^{N\times N}$ such that $\textbf{X}(t+T)=\textbf{X}(t)\tilde{\textbf{M}}$ for all $t \in \mathcal{R}$ and $\tilde{\textbf{M}}=\textbf{X}^{-1}(0)\textbf{X}(T)$ \cite{ floquet1883equations, gokccek2004stability}. Matrix $\tilde{\textbf{M}}$ is called the monodromy matrix.
\end{remark}
\begin{remark} If $\textbf{X}(0)=I$, then $\tilde{\textbf{M}}=\textbf{X}(T)$. 
\end{remark}
\begin{remark}
Because $\tilde{\textbf{M}}$ is non-singular, we obtain $\tilde{\textbf{M}}=e^{\textbf{B}T}$, or equivalently, $\textbf{B} = \ln(\tilde{\textbf{M}})/T$.
\end{remark}

Theorem~\ref{floquet_cont} implies that the solution of \eqref{linear_ode} is represented as 
\begin{equation}
\textbf{x}(t)=\textbf{X}(t)\textbf{x}(0)=\textbf{P}(t)e^{\textbf{B}t}\textbf{x}(0),
\label{floquet_solutuon}
\end{equation}
where $\textbf{x}(0)$ is the initial state.

Each eigenvalue $\lambda_i$ of the monodromy matrix is known as the Floquet or characteristic multiplier. The eigenvalues of $\textbf{B}$ are known as Floquet or characteristic exponents, and they are equal to $\ln(\lambda_{i})/ T$.  We define the largest Floquet exponent by
\begin{equation}
\lambda_\text{F} = \frac{\ln \left| \lambda_{\max} \right|}{T},
\end{equation}
where $\lambda_{\max}$ is the Floquet multiplier with the largest modulus. The largest Floquet exponent gives the average growth rate along the dominant eigenvector direction over one period. 
The equilibrium $\textbf{x}(t) = 0$ of the dynamical system given by Eq.~\eqref{linear_ode}, corresponding to the disease-free equilibrium in our context, is asymptotically stable if and only if $\lambda_{\text{F} }< 0$.

We are interested in $T$-periodic and switching ${\bf M}(t)$, that is, ${\bf M}(t)$ is constant for some duration, then it discontinuously switches to a different matrix that is used for another duration of time, and so on. In this case, the Floquet theory is simplified as follows.
\begin{thm} \label{floquet_piece_switch}[Floquet's theorem for linear dynamical systems with switching matrices \cite{gokccek2004stability}] 
Let $\textbf{M}(t)$ be a switching  $T$-periodic matrix given by
\begin{equation}
\textbf{M}(t)  = 
\begin{cases}
         \textbf{M}^{(1)} &  0\leq t < \tau_1, \\
        \textbf{M}^{(2)} & \tau_1 \le t < \tau_1 + \tau_2,\\
            \; \; \vdots & \quad \quad \; \vdots\\
            \textbf{M}^{(\ell)} & \tau_1 + \cdots + \tau_{\ell-1} \le t < T.
\end{cases}
\label{switching_network}
\end{equation}
Then, the monodromy matrix is represented as
\begin{equation}
\tilde{\textbf{M}} = e^{\tau_{\ell} \textbf{M}^{(\ell)}} e^{\tau_{\ell-1} \textbf{M}^{(\ell-1)}}
\cdots e^{\tau_1 \textbf{M}^{(1)}}.
\label{switching_monodromy}
\end{equation}
\end{thm}
\begin{remark}
Matrix $\tilde{\textbf{M}}$ for the deterministic SIS model on periodic switching networks is the same as $\mathcal{T}$ given by Eq.~\eqref{linearize_sismodel_conts_temporal_new}. 
\end{remark}


\section{Simulations of the stochastic SIS model and nonlinear IBA \label{appendix:gillespie}}

We used SBMs to generate the networks with two blocks, each of which represents a community and nodes within the same block have statistically the same connectivity patterns ~\cite{abbe2017community, lee2019review}. Each community has $N/2 = 1000$ nodes. Each edge $(i, j)$ exists with probability $[\textbf{A}^{(1)}]_{c_ic_j}/N$ for the first duration of one cycle and $[\textbf{A}^{(2)}]_{c_ic_j}/N$ for the second duration, where $c_i \in \{1, 2 \}$ is the block to which the $i$th node belongs. The $2\times 2$ matrix informing the SBM is the one used in Fig.~\ref{Floquet_multipliers_of_network}. In this manner, we obtain $2000 \times 2000$ adjacency matrices $\textbf{A}^{(1)}$ and $\textbf{A}^{(2)}$, and hence $\hat{\textbf{M}}$. 

\begin{figure}[t]
\begin{center}
\includegraphics[width = 1.01\linewidth]{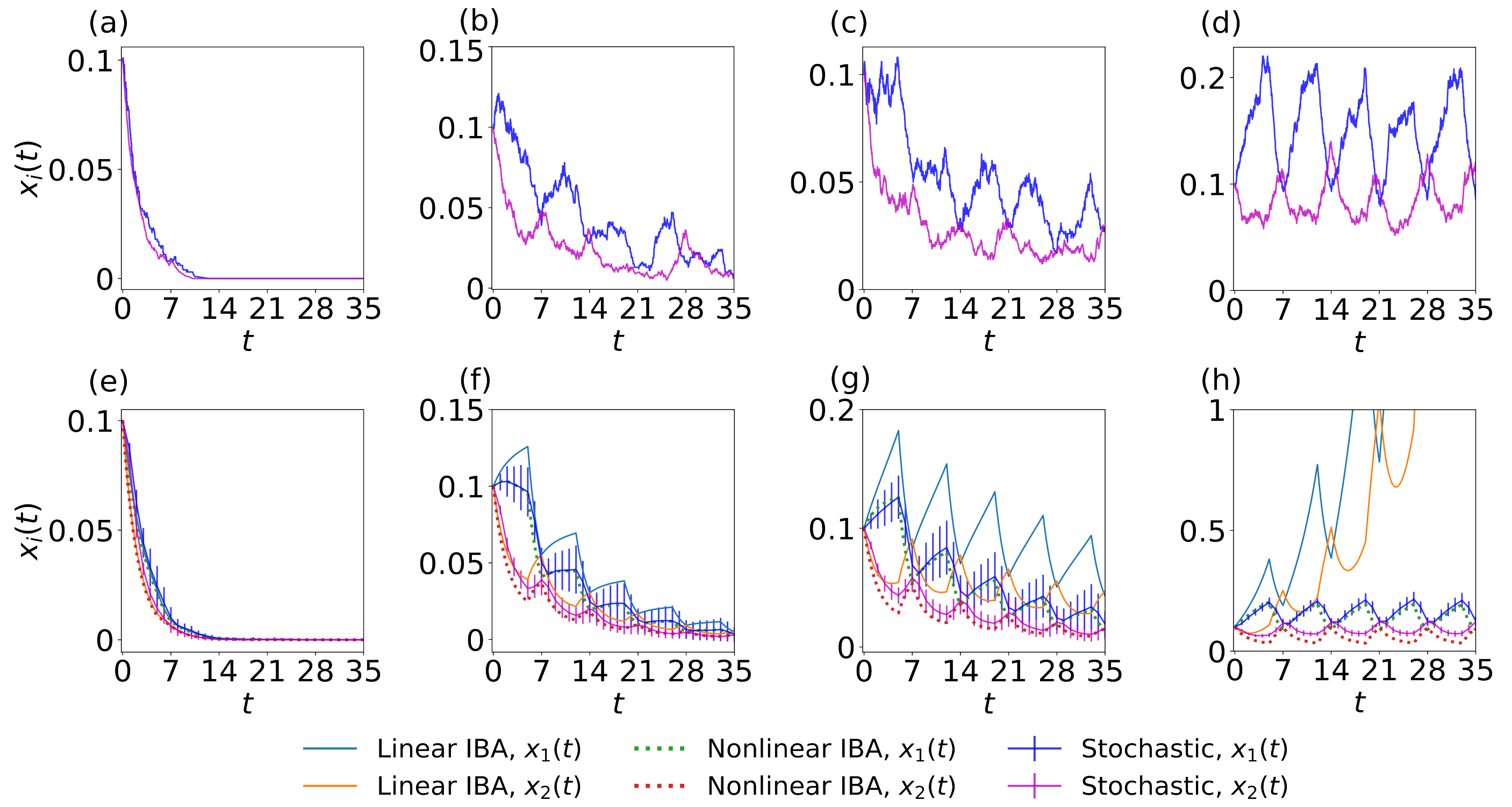}
     \caption{
 Agent-based simulations of the SIS model on periodic switching networks. (a)--(d) Results for a single run. (e)--(h) Comparison between the results of the agent-based simulations, linear IBA, and nonlinear IBA.
The value of $\mu$ $(=0.5)$, $\textbf{A}^{(1)}$, $\textbf{A}^{(2)}$, and the switching schedule (i.e., $\ell=2$, $T=7$, and $\tau_1 = 5$) are the same as those used in Fig.~\ref{Floquet_multipliers_of_network}. (a) and (e): $\beta=0.01$. (b) and (f): $\beta=0.028$. (c) and (g): $\beta = 0.032$. (d) and (h): $\beta = 0.04$. 
In (e)--(h), the lines and error bars for the agent-based simulations represent the mean and standard deviation on the basis of 100 simulations with the same initial condition. 
In all panels, we initially infected 10\% of nodes in each of the two communities, i.e., $(x_1(0), x_2(0)) = (0.1, 0.1)$.}
     \label{linear_nonlinear_SIS}
     \end{center}
\end{figure} 

We show in Fig.~\ref{linear_nonlinear_SIS}(a)--(d) the time course of $x_1(t)$ and $x_2(t)$ in a single run
of stochastic SIS dynamics simulated using the Gillespie algorithm.
We employed the four values of $\beta$ used in Fig.~\ref{Floquet_multipliers_of_network}(b)--(e).
The fluctuations in $x_1(t)$ and $x_2(t)$ are large because the results shown are from a single run of stochastic SIS dynamics. However, there is some indication of the Parrondo paradox and anti-phase oscillation in Fig.~\ref{linear_nonlinear_SIS}(b) and (c).
For a value of $\beta$ at which the paradox occurs according to our theoretical analysis (i.e., $\beta=0.032$; shown in Fig.~\ref{linear_nonlinear_SIS}(c)), we further carried out a single run of stochastic simulation on each of the two static networks with $N=2000$ nodes constituting the periodic switching network used in Fig.~\ref{linear_nonlinear_SIS}. Figure~\ref{stochastic_parrondo_static}(a) and (b) represents an example time course of $x_1(t)$ and $x_2(t)$ on the first and second static network, respectively. We find that $x_1(t)$ and $x_2(t)$ do not decay exponentially over time in each static network in this run.
Therefore, although not all the runs show this behavior, the Parrondo paradox can occur in a single run of stochastic simulation. 

\begin{figure}[t]
\begin{center}
\includegraphics[scale=0.3]{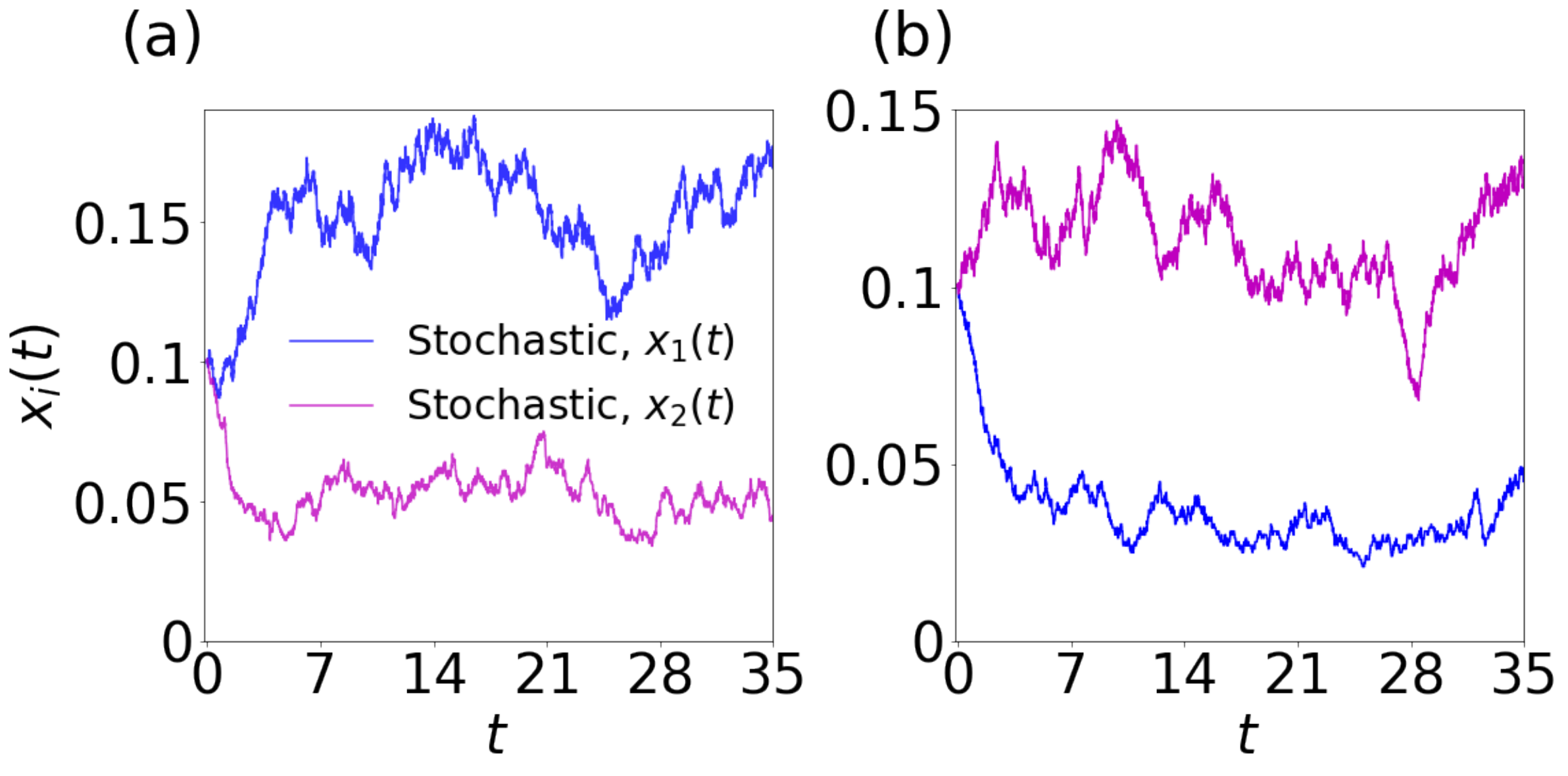}
     \caption{A single run of stochastic SIS simulation over the two static networks with $N=2000$ nodes generated by the SBM. The block adjacency matrices $\textbf{A}^{(1)}$ and $\textbf{A}^{(2)}$ are the same as those
used in Fig.~\ref{linear_nonlinear_SIS}. We set $\beta=0.032$, corresponding to Fig.~\ref{linear_nonlinear_SIS}(c) and (g), and $\mu=0.5$. (a) $\textbf{A}^{(1)}$. (b) $\textbf{A}^{(2)}$.}
     \label{stochastic_parrondo_static}
     \end{center}
\end{figure} 

We then simulate the same stochastic SIS dynamics 100 times for each value of $\beta$ and average $x_1(t)$ and $x_2(t)$ over the 100 runs. We show the results, together with the standard deviation of $x_1(t)$ and $x_2(t)$, in Fig.~\ref{linear_nonlinear_SIS}(e)--(h). We find a consistent tendency of the Parrondo paradox and anti-phase oscillation in the averaged time course of $x_1(t)$ and $x_2(t)$ (see Fig.~\ref{linear_nonlinear_SIS}(f) and (g)).

For comparison, we also show numerically simulated solutions of the nonlinear IBA, Eq.~\eqref{nonlinear_matrix_SIS}, and its linearized variant, Eq.~\eqref{linearizesismodel}, for the $2\times 2$ periodic switching network in the same figure. The time courses of $x_1(t)$ and $x_2(t)$ for the linearized IBA shown in Fig.~\ref{linear_nonlinear_SIS}(e)--(h) are the same as those shown in Fig.~\ref{Floquet_multipliers_of_network}(b)--(e). We observe that the nonlinear IBA shows a less pronounced Parrondo paradox and that the average time courses of the stochastic simulations when the paradox occurs are close to those for the nonlinear IBA (see Fig.~\ref{linear_nonlinear_SIS}(f) and (g)). We conclude that stochastic agent-based simulations also show the Parrondo paradox, albeit to a lesser extent than the linearized IBA with a small number of communities shown in the main text.

\section{Anti-phase oscillatory behavior with examples }\label{appendix:antiphase_example}

We measure $q$, i.e., the extent of the anti-phase oscillation defined in section~\ref{sec:decomposition}, for the eight two-node periodic switching networks that we used in Fig.~\ref{perturbation}. The upper subpanel in each of Fig.~\ref{fraction_of_timeseries}(a)--(h) shows $\lambda_{\text{F}}$, $\lambda_{\max}(\textbf{M}^{(1)})$, and $\lambda_{\max}(\textbf{M}^{(2)})$, and the range of $\beta$ in which the paradox occurs, similar to Fig.~\ref{Floquet_multipliers_of_network}.
We find that the first four networks showing the Parrondo paradox (i.e., Fig.~\ref{fraction_of_timeseries}(a)--(d)) produce large $q$ values (i.e., above approximately $0.8$) when $\beta$ is close to the epidemic threshold. The largest $q$ value attained in these four networks is $1$, indicating perfect anti-phase dynamics between $x_1(t)$ and $x_2(t)$. The range of $\beta$ in which $q=1$ roughly coincides with where the Parrondo paradox occurs. However, in one of the four periodic switching networks, shown in Fig.~\ref{fraction_of_timeseries}(c), the range of $\beta$ in which $q=1$ extends to larger values of $\beta$ far beyond where the paradox occurs. In contrast, for the other four networks that do not show the paradox (Fig.~\ref{fraction_of_timeseries}(e)--(h)), the largest $q$ value is substantially smaller than for the first four networks, while $q$ values are peaked near the epidemic thresholds for each momentarily static network, except in Fig.~\ref{fraction_of_timeseries}(g). 
Note that $q=0.5$ implies that the anti-phase behavior exists only half of the time so one cannot really regard this as anti-phase behavior. 
In Fig.~\ref{fraction_of_timeseries}(g), we find $q=0$, which is the complete absence of anti-phase dynamics, for all values of $\beta$.

\begin{figure}[t]
\begin{center}
\includegraphics[width = 1.01\linewidth]{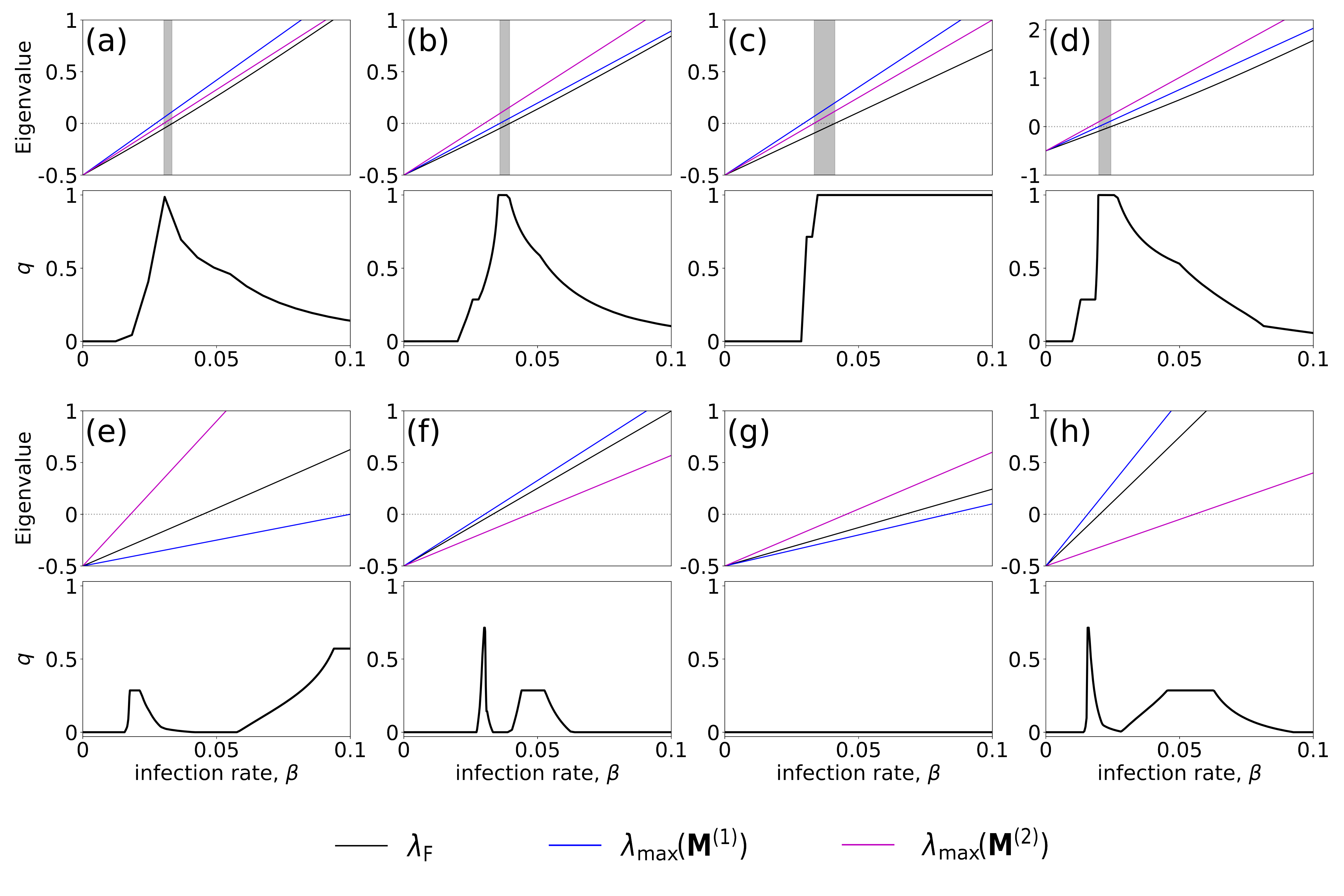}
\caption{Fraction of the time showing anti-phase behavior, $q$, in eight two-node periodic switching networks with $\ell = 2$, $T=7$, and $\tau_1 = 5$. The eight networks used are the same as those used in Fig.~\ref{perturbation}. We recall that the Parrondo paradox behavior occurs for the first four networks in a range of $\beta$ values, i.e., shaded regions, and it does not occur for the last four networks. We compare $q$ with $\lambda_{\text{F}}$, $\lambda_{\max}(\textbf{M}^{(1)})$, and $\lambda_{\max}(\textbf{M}^{(2)})$ as a function of the infection rate, $\beta$, for each periodic switching network. We set $\mu = 0.5$. The lines for $\lambda_{\text{F}}$, $\lambda_{\max}(\textbf{M}^{(1)})$, and $\lambda_{\max}(\textbf{M}^{(2)})$ are the same as those shown in Fig.~\ref{perturbation}. Observe that there is a strong association between the Parrondo paradox behavior and the appearance of large $q$ values.
}
\label{fraction_of_timeseries}
\end{center}
\end{figure}

\section*{Acknowledgments}

D.T. and N.M. acknowledge support from National Science Foundation (under Grant No. 2052720).  D.T. also acknowledges National Science Foundation award DMS-2401276 and Simon’s Foundation Grant 587333. N.M. also acknowledges support from the Japan Science and Technology Agency (JST) Moonshot R\&D (under Grant No. JPMJMS2021), the National Science Foundation (under Grant No. 2204936), and JSPS KAKENHI (under Grant No. JP 21H04595, 23H03414, 24K14840, W24K030130).

  \bibliography{references}

\end{document}